\documentclass{article}

\usepackage[american]{babel}
\usepackage[utf8]{inputenc}
\usepackage[T1]{fontenc}

\usepackage{listings}
\usepackage[most]{tcolorbox}
\usepackage{xcolor}

\definecolor{lightgray}{gray}{0.95}

\tcbset{
  enhanced jigsaw,
  colback=lightgray!5!white,
  colframe=lightgray,
  fonttitle=\bfseries
}

\newtcblisting{codebox}[2][]{%
  listing only,
  title=#2,
  listing options={
    basicstyle=\ttfamily\footnotesize,
    language=R,
    numbers=left,
    numberstyle=\tiny,
    stepnumber=1,
    numbersep=5pt,
    breaklines=true,
    showstringspaces=false,
    #1
  },
  colback=lightgray!30!white,  % darker background
  colframe=lightgray!50!black, % darker frame
  boxrule=0.5pt,
  sharp corners,
  fonttitle=\bfseries,
  enhanced
}

\usepackage[most]{tcolorbox}
\tcbset{enhanced jigsaw,colback=lightgray!5!white,colframe=lightgray,fonttitle=\bfseries}

\usepackage[english=american]{csquotes}
\usepackage{verbatim}
\usepackage{lmodern}

\usepackage{graphics}
\usepackage{multirow}
\usepackage{amsmath,amssymb,amsthm}
\usepackage[title]{appendix}
\usepackage{xcolor}

\usepackage{url}
\usepackage{mathtools}
\usepackage{marvosym} %Lightning
\usepackage{dsfont} %indicator function
\usepackage{enumerate}
\usepackage{float}
\usepackage[normalem]{ulem}
\usepackage[hidelinks]{hyperref}

\usepackage{fullpage}
\usepackage[font={small},format=plain,labelfont=bf]{caption}
\usepackage{authblk}
\usepackage[numbers,sort]{natbib}

\newcommand{\N}{\mathbb{N}} %natural numbers
\newcommand{\R}{\mathbb{R}} %real numbers
\newcommand{\BT}{\mathcal{BT}^{\ast}_n}
\newcommand{\T}{\mathcal{T}^{\ast}_n}
\newcommand{\id}{\textup{id}}
\newcommand{\bigcell}[2]{\begin{tabular}{@{}#1@{}}#2\end{tabular}}

\DeclareMathOperator*{\argmin}{arg\,min}

\theoremstyle{definition}
\newtheorem{Def}{Definition}[section]
\newtheorem{Theo}[Def]{Theorem}
\newtheorem{Prop}[Def]{Proposition}
\newtheorem{Lem}[Def]{Lemma}
\newtheorem{Cor}[Def]{Corollary}
\newtheorem{Rem}[Def]{Remark}

\newcommand\blfootnote[1]{%
  \begingroup
  \renewcommand\thefootnote{}\footnote{#1}%
  \addtocounter{footnote}{-1}%
  \endgroup
}

\makeatletter
\newcommand\opteq[1]{\mathrel{\mathpalette\opt@eq{#1}}}
\newcommand{\opt@eq}[2]{%
  \begingroup
  \sbox\z@{$#1#2$}%
  \sbox\tw@{\resizebox{!}{.5\ht\z@}{$\m@th#1($}}%
  \nonscript\hskip-\wd\tw@
  \mkern1mu
  \raisebox{-.35\ht\z@}[0pt][0pt]{\resizebox{!}{.5\ht\z@}{$\m@th#1($}}%
  \mkern-1mu
  {#2}%
  \mkern-1mu
  \raisebox{-.35\ht\z@}[0pt][0pt]{\resizebox{!}{.5\ht\z@}{$\m@th#1)$}}%
  \mkern1mu
  \nonscript\hskip-\wd\tw@
  \endgroup
}
\makeatother
\newcommand{\leoq}{\opteq{\leq}}
\newcommand{\geoq}{\opteq{\geq}}

\begin{document}
\setcounter{tocdepth}{4}
\setcounter{secnumdepth}{4}

\title{Metaconcepts of rooted tree balance}

%List of authors, with corresponding author marked by asterisk
\author[1,$\ast$]{Mareike Fischer}
\author[1]{Tom Niklas Hamann}
\author[2]{Kristina Wicke}

\affil[1]{Institute of Mathematics and Computer Science, University of Greifswald, Greifswald, Germany}
\affil[2]{Department of Mathematical Sciences, New Jersey Institute of Technology, Newark, NJ, USA}

\date{}
\maketitle

\begin{abstract}
Measures of tree balance play an important role in many different research areas such as mathematical phylogenetics or theoretical computer science. Typically, tree balance is quantified by a single number which is assigned to the tree by a balance or imbalance index, of which several exist in the literature. Most of these indices are based on structural aspects of tree shape, such as clade sizes or leaf depths. For instance, indices like the Sackin index, total cophenetic index, and $\widehat{s}$-shape statistic all quantify tree balance through clade sizes, albeit with different definitions and properties.

In this paper, we formalize the idea that many tree (im)balance indices are functions of similar underlying tree shape characteristics by introducing metaconcepts of tree balance. A metaconcept is a function $\Phi_f$ that depends on a function $f$ capturing some aspect of tree shape, such as balance values, clade sizes, or leaf depths. These metaconcepts encompass existing indices but also provide new means of measuring tree balance. The versatility and generality of metaconcepts allow for the systematic study of entire families of (im)balance indices, providing deeper insights that extend beyond index-by-index analysis.
\end{abstract}


\blfootnote{$^\ast$Corresponding author\\ \textit{Email address:} \url{mareike.fischer@uni-greifswald.de, email@mareikefischer.de}}

\section{Introduction}
The study of tree balance is an integral part of many different research areas. For example, tree balance is used in evolutionary biology and phylogenetics to study macroevolutionary processes such as speciation and extinction. Additionally, balanced trees are important in computer science, for instance, in the context of search trees~\cite{Andersson1993,Knuth1}.

In phylogenetics, the balance of a tree is usually quantified by so-called tree (im)balance indices. Intuitively, an (im)balance index is a function that assigns a single numerical value to a tree, assessing some aspect of its shape. The greater (smaller) the value, the more balanced the tree according to the respective (im)balance index. Over the last few decades, there has been a surge in the development of tree (im)balance indices, and numerous such indices are now available (for an overview and categorization see the recent survey by \citet{Fischer2023}). Despite the multitude of different (im)balance indices available, many of them employ similar underlying tree shape characteristics such as balance values, clade sizes, or leaf depths, albeit with different definitions and properties.

The main goal of this paper is to formalize the idea that many (im)balance indices for rooted trees are functions of similar underlying tree shape characteristics by introducing metaconcepts of tree balance. Here, a metaconcept is a function $\Phi_f$ that depends on a function $f$ capturing some aspect of tree shape, such as balance values, clade sizes, or leaf depths. This idea is inspired by a paper on unrooted trees: \citet{Fischer2021b} analyzed the balance of unrooted trees and introduced a measure $\Phi_f$, which can be regarded as a metaconcept for unrooted tree balance. Choosing $\Phi_f$ to be the sum function and $f$ to be a function of split sizes, the authors showed that this metaconcept leads to a family of functions suitable for measuring unrooted tree imbalance if $f$ is strictly increasing. Furthermore, very recently, \citet{Cleary2025} essentially used the idea of a metaconcept based on clade sizes to show that a wide range of clade-size based measures satisfying concavity and monotonicity conditions are minimized by the so-called complete or greedy from the bottom tree~\cite{Coronado2020a,Fill1996} and maximized by the so-called caterpillar tree (both trees are formally defined below).

In this paper, we provide a formal definition of a metaconcept for rooted trees. We then specialize this metaconcept to three classes of metaconcepts suitable to measure tree (im)balance. These metaconcepts are based on certain sequences that can be associated with rooted trees, namely the clade size sequence, the leaf depth sequence, and, in the case of binary trees, additionally the balance value sequence.
We rigorously study all metaconcepts, characterize which choices of the function $f$ lead to (im)balance indices, analyze extremal trees and values for the metaconcepts, and investigate further desirable properties such as locality and recursiveness.

\citet{Cleary2025} proved that the clade size metaconcept yields an imbalance index for binary trees if $f$ is strictly increasing and strictly concave. We extend this result to include functions that are strictly increasing and either strictly convex or affine. Moreover, the clade size metaconcept also defines an imbalance index for arbitrary trees if $f$ is additionally either $2$-positive, i.e., $f(x) > 0$ for all $x \geq 2$, in the concave and convex cases, or, in the case of affine functions, if $f$ has a non-negative intercept. Further, the leaf depth metaconcept yields an imbalance index for both arbitrary and binary trees if $f$ is strictly increasing and either convex or affine. Finally, the balance value metaconcept defines an imbalance index for binary trees for all strictly increasing functions $f$ without additional constraints.

To help users identify which imbalance index derived from a metaconcept best suits their specific aims, we provide four decision trees (Figures~\ref{Fig:decision_tree_first_min} and~\ref{Fig:decision_tree_first_meta}). These decision trees are based on three key properties of the index: (1) whether it applies to binary trees or to arbitrary trees; (2) the underlying structural aspect of the tree it captures (such as balance values, clade sizes, or leaf depths); and (3) the set of binary minimizing trees. For the third criterion, we give four possible options to choose from (cf. Figure \ref{Fig:min_trees_n12}). Additionally, we provide code for computing the metaconcepts using the \textsf{R} packages \texttt{treebalance}~\cite{Fischer2023} and \texttt{ape}~\cite{Paradis2019}.

We remark that our metaconcepts encompass various existing tree imbalance indices such as the Sackin and Colless indices, and we highlight these connections in the course of the paper. The power of our metaconcepts is that they are naturally more general and versatile than individual indices. Next to leading to new (im)balance indices, the metaconcepts thus also provide a new framework to study the properties (such as extremal trees and values) of whole families of existing imbalance indices holistically, rather than on an individual index basis.

The present manuscript is organized as follows: In Section~\ref{Sec:Preliminaries}, we present all definitions and notations needed throughout this manuscript and summarize some known results. Section~\ref{Sec:Results} then contains all our results: In Section~\ref{Subsec:Sequences}, we establish some general results on the underlying tree shape sequences employed in this paper. Section~\ref{Subsec:Metaconcepts} then discusses the resulting metaconcepts and their properties in depth. We start with the balance value metaconcept (Section~\ref{Sec:BVM}), then turn to the clade size metaconcept (Section~\ref{Sec:CSM}), and finally consider the leaf depth metaconcept (Section~\ref{Sec:LDM}). In Section~\ref{Sec:locality_recursiveness}, we analyze the locality and recursiveness of all three metaconcepts. We conclude our manuscript with a brief discussion and highlight some directions for future research in Section~\ref{Sec:Discussion}.

\section{Preliminaries}\label{Sec:Preliminaries}

In this section, we introduce all concepts relevant for the present manuscript. We start with some general definitions. We mainly follow the notation of~\cite{Fischer2023}.

\subsection{Definitions and notation} \label{sec:def}
\paragraph*{Rooted trees}
A \emph{rooted tree} (or simply \emph{tree}) is a directed graph $T = (V(T),E(T))$, with vertex set $V(T)$ and edge set $E(T)$, containing precisely one vertex of in-degree zero, the root (denoted by $\rho$), such that for every $v \in V(T)$ there exists a unique path from $\rho$ to $v$ and such that there are no vertices with out-degree one. We use $V_L(T) \subseteq V(T)$ to refer to the leaf set of $T$ (i.e., $V_L(T) = \{v \in V(T): \text{out-degree}(v)=0\}$), and we use $\mathring{V}(T)$ to denote the set of inner vertices of $T$ (i.e., $\mathring{V}(T) = V(T) \setminus V_L(T)$). Moreover, we use $n$ to denote the number of leaves of $T$, i.e., $n = \vert V_L(T) \vert$. Note that $\rho \in \mathring{V}(T)$ if $n \geq 2$. If $n=1$, $T$ consists of only one vertex, which is at the same time the root and its only leaf.

A rooted tree is called \emph{binary} if all inner vertices have out-degree two, and for every $n \in \mathbb{N}_{\geq 1}$, we denote by $\BT$ the set of (isomorphism classes of) rooted binary trees with $n$ leaves and by $\T$ the set of (isomorphism classes of) rooted trees with $n$ leaves. We often call a tree $T \in \T$ an \textit{arbitrary tree}, but remark that arbitrary trees are also sometimes referred to as non-binary trees in the literature (even though binary trees are also contained in the set of arbitrary trees).

\paragraph*{Depth and height}
The \emph{depth} $\delta_T(v)$ (or $\delta_v$ for brevity) of a vertex $v \in V(T)$ is the number of edges on the path from $\rho$ to $v$, and the height $h(T)$ of $T$ is the maximum depth of any leaf, i.e., $h(T) = \max_{x \in V_L(T)} \delta_T(x)$.

\paragraph*{Ancestors, descendants, and (attaching) cherries}
Let $u,v \in V(T)$ be vertices of $T$. Whenever there exists a path from $u$ to $v$ in $T$, we say that $u$ is an \emph{ancestor} of $v$ and $v$ is a \emph{descendant} of $u$. Note that this implies that each vertex is an ancestor and a descendant of itself. If $u$ and $v$ are connected by an edge, i.e., if $(u,v) \in E(T)$, we also say that $u$ is the \emph{parent} of $v$ and $v$ is a \emph{child} of $u$. The \emph{lowest common ancestor} $LCA_T(u,v)$ of two vertices $u,v \in V(T)$ is the unique common ancestor of $u$ and $v$ that is a descendant of every other common ancestor of them. Moreover, two leaves $x,y \in V_L(T)$ are said to form a \emph{cherry}, if they have the same parent, which is then also called a \emph{cherry parent}. Finally, by \emph{attaching a cherry} to a tree $T$ to obtain a tree $T'$, we mean replacing a leaf $x \in V_L(T)$ by a cherry. Notice that $T'$ has one more leaf than $T$.

\paragraph*{(Maximal) pending subtrees, clade sizes, and standard decomposition}
Given a tree $T$ and a vertex $v \in V(T)$, we denote by $T_v$ the \emph{pending subtree} of $T$ rooted in $v$ and use $n_T(v)$ (or $n_v$ for brevity) to denote the number of leaves in $T_v$, also called the \emph{clade size} of $v$. We will often decompose a rooted tree $T$ on $n \geq 2$ leaves into its maximal pending subtrees rooted in the children of $\rho$. We denote this decomposition as $T = (T_{v_1}, \ldots, T_{v_k})$, where $v_1, \ldots, v_k$ are the children of the root in $T$, and refer to it as the \emph{standard decomposition} of $T$. If $T$ is binary, we have $k=2$, and thus $T = (T_{v_1}, T_{v_2})$. Throughout, \emph{subtree} will always refer to a pending subtree.

\paragraph*{Balance values, (perfectly) balanced vertices, and cophenetic values}
Now let $T$ be a rooted binary tree and let $v \in \mathring{V}(T)$ be an inner vertex of $T$ with children $v_1$ and $v_2$. The \emph{balance value} $b_T(v)$ (or $b_v$ for brevity) of $v$ is defined as $b_{T}(v) \coloneqq |n_{v_1} - n_{v_2}|$. An inner vertex $v$ is called \emph{balanced} if it fulfills $b_T(v) \leq 1$ and \emph{perfectly balanced} if $b_T(v) = 0$. Now, let $x,y \in V_L(T)$ be two leaves of a tree $T \in \T$. Then, the \emph{cophenetic value} of $x$ and $y$ is defined as $\varphi_T(x,y) \coloneqq \delta_T\left(LCA_T(x,y)\right)$, i.e., it is the depth of their lowest common ancestor.

\begin{figure}[ht]
\centering
	\includegraphics[scale=2]{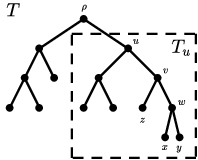}
	\caption{Rooted binary tree $T$ with eight leaves and root $\rho$. The vertices $\rho$, $u$, and $v$ are ancestors of $v$. The parent of $v$ is $u$ and $v$ is one of two children of $u$. The descendants of $v$ are $v$, $z$, $w$, $x$, and $y$. The lowest common ancestor of $x$ and $z$ is $LCA_T(x,z) = v$. The leaves $x$ and $y$ form a cherry whose parent is $w$. The pending subtree of $u$ is $T_u$, which is also one of the two maximal pending subtrees of $T$. It has five leaves and thus $n_T(u) = 5$, i.e., the clade size of $u$ is five. The balance value of $u$ and $v$ is one, i.e., $b_T(u) = b_T(v) = 1$, hence $u$ and $v$ are balanced. The vertex $w$ is balanced, too, and further, it is perfectly balanced, because $b_T(w) = 0$. The root $\rho$ is not balanced as $b_T(\rho) = 2$.}
	\label{Fig:BeispieleDefinitionen}
\end{figure}

\paragraph*{Important (families of) trees}
Next, we introduce some specific families of trees that will be important throughout this manuscript (see Figure~\ref{Fig:special_trees} for examples).

First, the \emph{maximally balanced tree} (or mb-tree for brevity), denoted by $T^{mb}_n$, is the rooted binary tree with $n$ leaves in which all inner vertices are balanced. Recursively, a rooted binary tree with $n \geq 2$ leaves is maximally balanced if its root is balanced and its two maximal pending subtrees are maximally balanced.

Second, the \emph{greedy from the bottom tree} (or gfb-tree for brevity), denoted by $T^{gfb}_n$, is the rooted binary tree with $n$ leaves that results from greedily clustering trees of minimal leaf numbers, starting with $n$ single vertices and proceeding until only one tree is left as described by~\cite[Algorithm 2]{Coronado2020a}.

Third, the \emph{fully balanced tree of height $h$} (or fb-tree for brevity), denoted by $T^{fb}_h$ is the rooted binary tree with $n=2^h$ leaves with $h \in \mathbb{N}_{\geq 0}$, in which all leaves have depth precisely $h$. Note that for $h \geq 1$, we have $T^{fb}_h = \left(T^{fb}_{h-1}, T^{fb}_{h-1}\right)$. Moreover, for $h \in \mathbb{N}_{\geq 0}$, $T^{fb}_h = T^{mb}_{2^h} = T^{gfb}_{2^h}$.

Fourth, the \emph{caterpillar tree} (or simply \emph{caterpillar}), denoted by $T^{cat}_n$, is the rooted binary tree with $n$ leaves that fulfills either $n=1$, or $n \geq 2$ and additionally has exactly one cherry.

Finally, the \emph{star tree}, denoted by $T^{star}_n$, is the rooted tree with $n$ leaves that either satisfies $n=1$, or $n \geq 2$ and additionally has a single inner vertex (the root), which is adjacent to all leaves.

Notice that all trees introduced above are unique (up to isomorphism) and have the property that all their pending subtrees are (smaller) trees of the same type. Moreover, we remark that the caterpillar is generally regarded as the most unbalanced (binary) tree, whereas the fully balanced tree is considered the most balanced binary tree when it exists, i.e., for leaf numbers that are powers of two. For other leaf numbers, both the maximally balanced tree and the greedy from the bottom tree are often regarded as the most balanced binary trees, whereas the star tree is usually considered to be the most balanced arbitrary tree.

\begin{figure}[ht]
\centering
	\includegraphics[width=0.7\textwidth]{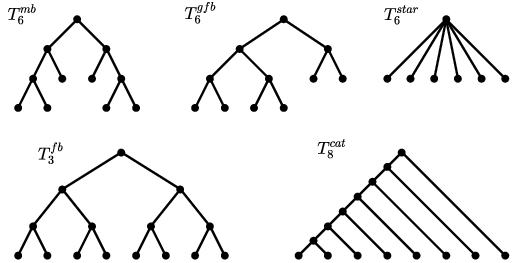}
	\caption{Examples of the special trees considered throughout this manuscript.}
	\label{Fig:special_trees}
\end{figure}

\paragraph*{Imbalance index, locality and recursiveness}
We next introduce the concept of a tree imbalance index. First, following~\citet{Fischer2023}, a \emph{(binary) tree shape statistic} is a function $t: \T(\BT) \rightarrow \R$ that depends only on the shape of $T$ but not on the labeling of vertices or the length of edges. Based on this, a tree imbalance index is defined as follows:
\begin{Def}[(Binary) imbalance index~(\citet{Fischer2023})]
    A (binary) tree shape statistic $t$ is called an \emph{imbalance index} if and only if
	\begin{enumerate}[(i)]
		\item the caterpillar $T^{cat}_n$ is the unique tree maximizing $t$ on its domain $\T(\BT)$ for all $n \geq 1$,
		\item the fully balanced tree $T^{fb}_h$ is the unique tree minimizing $t$ on $\BT$ for all $n = 2^h$ with $h \in \N_{\geq 0}$.
	\end{enumerate}
    If the domain of $t$ is $\BT$, we often call $t$ a \textit{binary imbalance index} to highlight this fact.
\end{Def}

Given two trees, say $T,T' \in \T(\BT)$, and an imbalance index, say $t$, we say that \emph{$T$ is more balanced than $T'$} (with respect to $t$) if $t(T) < t(T')$. More generally, when we say that a tree $T$ \emph{minimizes an imbalance index}, we mean that it minimizes it among all trees with the same leaf number as $T$. Analogously, when we compare a tree $T$ to a family of trees (such as the ones defined above), we always compare it to the family's representative that has the same number of leaves.

We note that in addition to imbalance indices, balance indices also exist. A balance index is minimized by the caterpillar tree and maximized by the fb-tree. Since a balance index can be obtained from an imbalance index (and vice versa) by multiplying by $-1$, and given that the majority of known indices are formulated as measures of imbalance, we focus exclusively on imbalance indices in this work.

Further, two imbalance indices $\varphi_1$ and $\varphi_2$ are considered \emph{equivalent}, if for all trees $T_1, T_2 \in \T(\BT)$, the following holds: $\varphi_1(T_1) < \varphi_1(T_2) \Longleftrightarrow \varphi_2(T_1) < \varphi_2(T_2)$. In other words, equivalence means that $\varphi_1$ and $\varphi_2$ rank trees in the same order from most balanced to least balanced.
\smallskip

We next turn to two desirable properties of imbalance indices, namely locality and recursiveness.

\begin{Def}[Locality~(\citet{Mir2013,Fischer2023})]
    Let $T \in \T(\BT)$ be a tree and let $v$ be a vertex of $T$. Further, let $T'$ be obtained from $T$ by replacing the subtree $T_v$ rooted in $v$ by a (binary) tree $T'_v$ with the same leaf number and also rooted in $v$. An imbalance index $t$ is called \emph{local} if it fulfills
    \[t(T) - t(T') = t(T_v) - t\left(T'_v\right) \text{ for all } v \in V(T).\]
\end{Def}

In other words, if $t$ is local and two trees $T$ and $T'$ differ only in a pending subtree, then the differences of their $t$-values is equal to the differences of the subtrees' $t$-values.
\smallskip
 
Next, we introduce the recursiveness of a tree shape statistic.

\begin{Def}[Recursiveness (based on~\citet{Fischer2023})]
\label{Def:recursivenessBook}
    A \emph{recursive tree shape statistic} of length $x \in \N_{\geq 1}$ is an ordered pair $(\lambda,r)$, where $\lambda \in \R^{x}$ and $r$ is an $x$-vector of symmetric functions each mapping a multiset of $x$-vectors to $\R$. In this definition, $x$ is the number of recursions that are used to calculate the index, the vector $\lambda$ contains the start value for each of the $x$ recursions, i.e., the values of $T \in \mathcal{T}^{\ast}_1$ if $n = 1$, and the vector $r$ contains the recursions themselves. In particular, $r_i(T) = \lambda_i$ for $n = 1$, and for $T = (T_1, \ldots, T_k)$, recursion $r_i$ operates on $k$ vectors of length $x$, namely $(r_1(T_1), \ldots, r_x(T_1)), \ldots, (r_1(T_k), \ldots, r_x(T_k))$, each representing one of the maximal pending subtrees $T_1, \ldots, T_k$ and containing their respective values. The recursions are symmetrical functions, i.e., the order of those $k$ vectors is permutable, because we are solely considering unordered trees. If only binary trees are considered, i.e., $k = 2$ for every pending subtree, we use the term \emph{binary recursive tree shape statistic}.
\end{Def}

In the following, we introduce our main concepts: the definition of a general metaconcept, three tree shape sequences, and three classes of metaconcepts -- each based on one of these sequences. We begin by defining the sequences.

\paragraph*{Balance value sequence, clade size sequence, and leaf depth sequence}
First, the \emph{balance value sequence} of a binary tree $T \in \BT$ is the list of balance values of all its inner vertices, arranged in ascending order. We denote this sequence by $\mathcal{B}\left(T\right) \coloneqq (b_1, \ldots, b_{n-1})$. The $i$-th entry of $\mathcal{B}(T)$ is denoted by $\mathcal{B}\left(T\right)_i$. Note that for any $T \in \BT$, the length of $\mathcal{B}(T)$ is $n-1$. Also note that $B(T) = (0, \ldots, 0)$ if and only if $T = T^{fb}_h$ (for a formal argument, see \cite[Corollary 1]{Coronado2020a}).

Second, the \textit{clade size sequence} of a tree $T \in \T$ is the list of clade sizes of all its inner vertices, arranged in ascending order. We denote this sequence by $\mathcal{N}(T) \coloneqq (n_1, \ldots, n_{|\mathring{V}(T)|})$, where $\mathcal{N}(T)_i$ represents the $i$-th entry of $\mathcal{N}(T)$. The length of the clade size sequence for a tree with $n \geq 2$ leaves can range from $1$ to $n-1$. Specifically, the sequence has length 1 if and only if $T$ is a star tree, and it has length $n-1$ if and only if $T$ is binary.

Third, the \textit{leaf depth sequence} of a tree $T \in \T$ is the list of leaf depths of all its leaves, arranged in ascending order. We denote this sequence by $\Delta(T) \coloneqq (\delta_1, \ldots, \delta_{n})$, where $\Delta(T)_i$ represents the $i$-th entry of $\Delta(T)$. Unlike the clade size sequence, the leaf depth sequence has always length $n$, regardless of whether the tree is binary.

\paragraph*{Balance value metaconcept, clade size metaconcept, and leaf depth metaconcept}
Next, we define the general metaconcept. In a second step, we derive three classes of metaconcepts from this definition, each based on one of the previously introduced sequences. Let $T \in \mathcal{T} \subseteq \T$ be a tree, and let $Seq(T)$ be a vertex value sequence on a subset $V' \subseteq V(T)$, i.e., a sequence that assigns each vertex $v \in V'$ a value $s_v$ derived from $v$. Assume that $Seq(T)$ is sorted in ascending order, and let $Seq(T)_i$ denote its $i$-th entry.

Furthermore, let $\omega \in \N_{\geq 1}$, $c = \min\limits_{T' \in \mathcal{T}}\left\{Seq(T')_1\right\}$, and $f: \R_{\geq c} \times \R^{\omega-1} \rightarrow \R$ be a function that depends on an entry of $Seq(T)$ and $\omega-1$ additional values $o_1(T), \ldots, o_{\omega-1}(T)$, such as the number of inner vertices, i.e., $o_i(T) = |\mathring{V}(T)|$, or the number of leaves of $T$, i.e., $o_i(T) = n$. Then,
\begin{align*}
    \Phi^{Seq}_{f}(T) &\coloneqq \sum\limits_{s \in Seq(T)} f(s,o_1(T), \ldots, o_{\omega-1}(T))
\end{align*}
is called the \textit{imbalance index metaconcept of order $\omega$}. Clearly,
\begin{align*}
    \Phi^{Seq}_{f}(T)
    = \sum\limits_{v \in V'} f(s_v,o_1(T), \ldots, o_{\omega-1}(T))
    = \sum\limits_{i = 1}^{|V'|} f(Seq(T)_i,o_1(T), \ldots, o_{\omega-1}(T)).
\end{align*}

Examples for known balance indices and their interpretations in the framework of metaconcepts of order $\omega$ can be found in Tables \ref{Tab:imbalance_indices_T} and \ref{Tab:imbalance_indices_BT}.

\noindent We now specialize the general metaconcept to three subclasses.
First, let $T \in \BT$ be a binary tree, and let $Seq(T) = \mathcal{B}(T)$ be its balance value sequence. Then, the \emph{balance value metaconcept} (BVM) \emph{of order $\omega$} is defined as
\[\Phi^{\mathcal{B}}_{f}\left(T\right) \coloneqq \sum\limits_{b \in \mathcal{B}\left(T\right)} f(b,o_1(T), \ldots, o_{\omega-1}(T)).\]
Second, if $T \in \T$ is a rooted tree, and $Seq(T) \in \{\mathcal{N}(T), \Delta(T)\}$, we obtain the \emph{clade size metaconcept} (CSM) \emph{of order $\omega$} defined as
\[\Phi^{\mathcal{N}}_{f}(T) \coloneqq \sum\limits_{n_v \in \mathcal{N}(T)} f(n_v,o_1(T), \ldots, o_{\omega-1}(T))\]
and the \emph{leaf depth metaconcept} (LDM) \emph{of order $\omega$} defined as
\[\Phi^{\Delta}_{f}(T) \coloneqq \sum\limits_{\delta \in \Delta(T)} f(\delta,o_1(T), \ldots, o_{\omega-1}(T)).\]

Note that for $n \geq 2$ the minimal balance value is $0$, the minimal clade size is $2$, and the minimal leaf depth is $1$. Hence, the value $c$ in the definition of the metaconcept equals the respective value.

Note that the clade size metaconcept of order $1$, when applied with a strictly increasing and strictly concave function $f$, corresponds to the function $\Phi_f$ in \citet{Cleary2025}. Recall that a (strictly) concave function satisfies $f(\lambda x + (1-\lambda)y) \geoq \lambda f(x) + (1-\lambda)f(y)$ for all $\lambda \in (0,1)$ and all $x,y \in \R$ and $x \neq y$. When choosing $\lambda = \frac{1}{2}$ and $y = x+2$, this yields the inequality $2 \cdot f(x) \geoq f(x-1) + f(x+1)$ and hence $f(x) - f(x-1) \geoq f(x+1) - f(x)$, i.e., the increments (strictly) decrease. Finally, we will sometimes use the fact that a differentiable function $f$ is (strictly) concave on an interval if and only if its derivative function $f'$ is (strictly) decreasing on that interval. Conversely, a (strictly) convex function has (strictly) increasing increments, meaning the inequalities are reversed. Moreover, we call a function $f: \R_{\geq 0} \rightarrow \R$ $2$-positive, if $f(x) > 0$ for all $x \geq 2$ and non-negative if $f(x) \geq 0$ for all $x \geq 0$. In the case of an affine function $f(x) = m \cdot x + a$, we refer to $m$ as the slope and $a$ as the intercept.

\subsection{Known imbalance indices}
Tables~\ref{Tab:imbalance_indices_T} and~\ref{Tab:imbalance_indices_BT} define various known imbalance indices. Note that the choice of logarithm base is arbitrary. Additionally, we follow the conventions that $\frac{0}{0} = 0$ and that a sum over an empty set equals zero.
Note that \citet{Fischer2023} demonstrated that all functions listed in Table \ref{Tab:imbalance_indices_T} satisfy the definition of an imbalance index on $\T$, except for the $\widehat{s}$-shape statistic, which is only a binary imbalance index. Moreover, all functions listed in Table \ref{Tab:imbalance_indices_BT} are binary imbalance indices, too.

\begin{table}[htbp]
    \small
    \centering
    \caption{Definitions of imbalance indices that are applicable to arbitrary trees, i.e., those with domain $\T$ (notice that the $\widehat{s}$-shape statistic is applicable to arbitrary trees, but is only an imbalance index on $\BT$). It is straightforward to see that these imbalance indices are induced by the clade size metaconcept $\Phi^{\mathcal{N}}_f$ and the leaf depth metaconcept $\Phi^{\Delta}_f$, respectively, when the function $f$ is chosen as specified in the two rightmost columns. The Sackin index and the $\widehat{s}$-shape statistic are induced by the first-order metaconcept with $\id$ as the identity function. In contrast, the average leaf depth is induced by the second-order metaconcept. Moreover, the total cophenetic index is induced on $\BT$ by the second-order and on $\T$ by the third-order metaconcept. This is because for binary trees we have $|\mathring{V}(T)| = n-1$, so no further additional value than $n$ is needed. For further details on the total cophenetic index, see Remark \ref{Rem:clade_meta_eqiuv_Sackin}.}
    \begin{tabular}{cc|cc}
        imbalance index & definition & CSM $\Phi^{\mathcal{N}}_f$ & LDM $\Phi^{\Delta}_f$\\
    \hline
    & &\\
        \bigcell{c}{\textbf{Sackin index} \\ \cite{Sackin1972,Shao1990,Blum2005,Fischer2021a}} & \bigcell{c}{$S(T) \coloneqq \sum\limits_{x \in V_L(T)} \delta_{T}(x)$ \\ $\phantom{S(T)} = \sum\limits_{v \in \mathring{V}(T)} n_{T}(v)$} & $\id$ & $\id$\\
    & &\\

        \bigcell{c}{\textbf{Average leaf depth} \\ \cite{Kirkpatrick1993}} & \bigcell{c}{$\overline{N}(T) \coloneqq \frac{1}{n} \sum\limits_{x \in V_L(T)} \delta_{T}(x)$ \\ $\phantom{\overline{N}(T)} = \frac{1}{n} \cdot S(T)$} & $f_{\overline{N}}(n_v,n) = \frac{1}{n} \cdot n_v$ & $f_{\overline{N}}(\delta,n) = \frac{1}{n} \cdot \delta$\\
    & &\\

        \bigcell{c}{\textbf{$\widehat{s}$-shape statistic} \\ \cite{blum2006c}} & $\widehat{s}(T) \coloneqq \sum\limits_{v \in \mathring{V}(T)} \log\left(n_T(v) -1\right)$ & $f_{\widehat{s}}(n_v) = \log(n_v -1)$ & --\\
    & &\\

        \bigcell{c}{\textbf{Total cophenetic index} \\ \cite{Mir2013}} & \bigcell{c}{$\Phi(T) \coloneqq \sum\limits_{\substack{(x,y) \in V_L(T)^2 \\ x \neq y}} \varphi_T(x,y)$\\ $\phantom{\Phi(T)} = \sum\limits_{v \in \mathring{V}(T) \setminus \left\{\rho\right\}} \binom{n_T(v)}{2}$} & $f_{\Phi}\left(n_v,n,|\mathring{V}(T)|\right) = \binom{n_v}{2} - \frac{\binom{n}{2}}{|\mathring{V}(T)|}$ & --
    \end{tabular}
    \label{Tab:imbalance_indices_T}
\end{table}

\begin{table}[htbp]
    \small
    \centering
    \caption{Definitions of binary imbalance indices, which are only applicable to binary trees, i.e., those with domain $\BT$. It is straightforward to see that these binary imbalance indices are induced by the balance value metaconcept $\Phi^{\mathcal{B}}_f$ when the function $f$ is chosen as specified in the right column. The (quadratic) Colless index is induced by the first-order metaconcept where $\id$ is the identity function, while the corrected Colless index is induced by the second-order metaconcept.}
    \begin{tabular}{cc|c}
        binary imbalance index & definition & BVM $\Phi^{\mathcal{B}}_f$\\
    \hline
    & &\\
        \bigcell{c}{\textbf{Colless index} \\ \cite{Colless1982,Shao1990}} & $C\left(T\right) \coloneqq \sum\limits_{v \in \mathring{V}\left(T\right)} b_{T}(v)$ & $\id$\\
    & &\\

        \bigcell{c}{\textbf{Corrected Colless index} \\ \cite{Heard1992}} & \bigcell{c}{$I_C\left(T\right) \coloneqq \frac{2}{(n-1)(n-2)} \cdot C\left(T\right)$\\ $I_C\left(T\right) \coloneqq 0$ for $n=1,2$} & $f_{I_C}(b,n) = \frac{2}{(n-1)(n-2)} \cdot b$\\
    & &\\

        \bigcell{c}{\textbf{Quadratic Colless index} \\ \cite{Bartoszek2021}} & $QC\left(T\right) \coloneqq \sum\limits_{v \in \mathring{V}\left(T\right)} b_{T}(v)^2$ & $f_{QC}(b) = b^2$
    \end{tabular}
    \label{Tab:imbalance_indices_BT}
\end{table}

\subsection{Known results}
Before presenting our new results, we first recall some previously established findings. We summarize key results concerning special trees as well as known imbalance indices.

\begin{Prop}[\citet{Coronado2020a}, Theorem 1 and Proposition 6]
\label{Prop:gfb_mb_min_Colless}
    Let $n \in \mathbb{N}_{\geq 1}$. The mb-tree $T^{mb}_n$ and the gfb-tree $T^{gfb}_n$ minimize the Colless index on $\BT$.
\end{Prop}

\begin{Rem}
\label{Rem:colless-minima}
    Notice that for most leaf numbers $n$, there are trees distinct from the mb-tree $T^{mb}_n$ and the gfb-tree $T^{gfb}_n$ that also minimize the Colless index. However, all binary trees with $n$ leaves minimizing the Colless index have been completely characterized by \citet[Proposition 1 and Proposition 3]{Coronado2020a}.
\end{Rem}

\begin{Lem}[\citet{Fischer2021a}, Theorem 2]
\label{Lem:Sackin_mintree_min_value}
    Let $T \in \T$ with $h_n = \lceil\log_2(n)\rceil$. Then, $T$ minimizes the Sackin index on $\BT$ if and only if $T = T^{fb}_{h_n}$ or $T$ employs precisely two leaf depths, namely $h_n-1$ and $h_n$. Moreover, in this case, $S(T) = -2^{h_n} + n \cdot (h_n+1)$, which equals $h_n \cdot 2^{h_n}$ if $n = 2^{h_n}$.
\end{Lem}

\begin{Rem}
\label{Rem:gfb_mb_from_fb_min_Sackin}
    Note that trees with $n$ leaves minimizing the Sackin index as characterized in the lemma above are precisely those trees that can be constructed from the fb-tree of height $h_n-1$ by attaching $n - 2^{h_n-1}$  cherries to its leaves. In particular, the gfb-tree and the mb-tree can be constructed in this way. This follows from the fact that the gfb-tree and the mb-tree minimize the Colless index on $\BT$ (Proposition \ref{Prop:gfb_mb_min_Colless}) and the fact that all trees minimizing the Colless index also minimize the Sackin index on $\BT$ (\citet[Proposition 9]{Coronado2020a}). Note that to construct the gfb-tree, one has to attach the cherries to the fb-tree from left to right (or vice versa) (\citet[Lemma 4.17]{Cleary2025}).
\end{Rem}

\begin{Prop}
\label{Prop:Min_Colless}[\citet{Coronado2020a}, Theorem 3]
    For every $n \in \mathbb{N}_{\geq 1}$, let $h_n \coloneqq \lceil\log_2(n)\rceil$. Then,
   \[ c_n:=\sum\limits_{i = 1}^{h_n -1} 2^i \cdot s(2^{-i} \cdot n)\]
    is the minimum value of the Colless index on $\BT$, where $s(x)$ is the distance from $x \in \R$ to its nearest integer,  i.e., $s = \min\limits_{z \in \mathbb{Z}} |x-z|$.
\end{Prop}

\begin{Lem}[\citet{Fischer2021a}, Theorem 1]
\label{Lem:Sackin_cat}
    The caterpillar uniquely maximizes the Sackin index on $\BT$, and we have $S\left(T^{cat}_n\right) = \frac{n \cdot (n+1)}{2} -1$.
\end{Lem}

\begin{Prop}
\label{Prop:gfb_min_clade_meta_concave_imb_index}[adapted from \citet{Cleary2025}, Corollary 4.4]
    Let $f$ be strictly increasing and strictly concave. Then, the clade size metaconcept $\Phi^{\mathcal{N}}_f$ is a binary imbalance index. Moreover, the gfb-tree $T^{gfb}_n$ uniquely minimizes the clade size metaconcept on $\BT$.
\end{Prop}

\begin{Prop}
\label{Prop:cat_max_clade_meta_str_incr}[adapted from \citet{Cleary2025}, Theorem 4.3]
    Let $f$ be strictly increasing. Then, the caterpillar $T^{cat}_n$ uniquely maximizes the clade size metaconcept $\Phi^{\mathcal{N}}_f$ on $\BT$.
\end{Prop}

Finally, we recall a result from \citet{Cleary2025} regarding the number of subtrees of the gfb-tree $T^{gfb}_n$ for all possible subtree sizes.

\begin{Theo}[\citet{Cleary2025}, Theorem 4.12] \label{theo:gfb-subtree-sizes}
    Let $n \geq 1$, and let $gfb_n(i)$ denote the number of subtrees of $T^{gfb}_n$ of size $i$ for $i=1, \ldots, n$. Let $h_i = \lceil \log(i) \rceil$. Then, we have:
    \[gfb_n(i) = \begin{cases}
        \left\lfloor\frac{n}{i}\right\rfloor & \text{if } i = 2^{h_i} \text{ and if } \left((n \mod i) = 0 \text{ or } (n \mod i) \geq 2^{h_i -1}\right),\\
        \left\lfloor\frac{n}{i}\right\rfloor -1 & \text{if } i = 2^{h_i} \text{ and if } \left(0 < (n \mod i) < 2^{h_i -1}\right),\\
        1 &\text{if } i \neq 2^{h_i} \text{ and } \left((n-i) \mod 2^{h_i -1}\right) = 0,\\
        0 &\text{if } i \neq 2^{h_i} \text{ and } \left((n-i) \mod 2^{h_i -1}\right) > 0.
    \end{cases}\]
\end{Theo}

We are now in the position to state our new results.

\section{Results} \label{Sec:Results} 
This section is divided into two subsections. The first subsection examines the underlying tree shape sequences of the metaconcepts, highlighting their differences and similarities. The second subsection analyzes each metaconcept in terms of its minimizing and maximizing trees, as well as its minimum and maximum values. Finally, we investigate their locality and recursiveness.

\subsection{Tree shape sequences} \label{Subsec:Sequences}
In this subsection, we analyze and compare three sequences derived from a (binary) tree: the balance value sequence, the clade size sequence, and the leaf depth sequence. Each of these sequences serves as the foundation for a specific metaconcept.

A shared property of the three sequences associated with a rooted tree is that they can be computed recursively. We will exploit this property to analyze the recursiveness of our metaconcepts.

To formalize the recursive structure of these sequences, we introduce an operator that allows us to merge two sequences while preserving ascending order. Let $Seq_1$  and $Seq_2$ be two sequences of lengths $n_1$ and $n_2$, respectively. We define their ordered union as $Seq_1 \overset{\rightarrow}{\cup} Seq_2 \coloneqq Seq$, where $Seq$ is a sequence of length $n_1 + n_2$ containing all elements of $Seq_1$ and $Seq_2$ arranged in ascending order. For example, $(1,4,5,13) \overset{\rightarrow}{\cup} (2,2,4,7,8) = (1,2,2,4,4,5,7,8,13)$. Additionally, for $a \in \N$, we define $Seq_1 + a$ as the sequence obtained by increasing each element of  $Seq_1$ by $a$. For example, $(1,4,5,13) +1 = (2,5,6,14)$.

\begin{Rem}
\label{Rem:recursiveness_sequences}
    Let $T \in \BT$ be a binary tree with standard decomposition $T = \left(T_1, T_2\right)$ such that $T_1$ and $T_2$ have $n_1$ and $n_2$ leaves, respectively. Notice that all inner vertices of a maximal pending subtree $T_i$ have the same balance value in $T_i$ as in $T$ and the root of $T$ has balance value $\left|n_1 - n_2\right|$. Hence,
    \[\mathcal{B}\left(T\right) = \mathcal{B}(T_1) \overset{\rightarrow}{\cup} \mathcal{B}(T_2) \overset{\rightarrow}{\cup} (|n_1 - n_2|).\]
    Now, let $T \in \T$ be an arbitrary tree with standard decomposition $T = (T_1, \ldots, T_k)$ such that the maximal pending subtree $T_i$ has $n_i$ leaves. Note that all inner vertices of a maximal pending subtree $T_i$ have the same clade size in $T_i$ as in $T$ and the root of $T$ has clade size $n_1 + \ldots + n_k = n$. Also note that the depth of a leaf in a maximal pending subtree $T_i$ is one less than in $T$. Thus,
    \[\mathcal{N}(T) = \mathcal{N}(T_1) \overset{\rightarrow}{\cup} \ldots \overset{\rightarrow}{\cup} \mathcal{N}(T_k) \overset{\rightarrow}{\cup} \underbrace{(n_1 + \ldots + n_k)}_{=n}\]
    and
    \[\Delta(T) = \left(\Delta(T_1) \overset{\rightarrow}{\cup} \ldots \overset{\rightarrow}{\cup} \Delta(T_k)\right) +1.\]
\end{Rem}

Another shared property of these sequences is that none of them uniquely characterize a (binary) tree. That is, two non-isomorphic (binary) trees can have the same sequence (see Figures \ref{Fig:6_5_6} and \ref{Fig:9_42_44} in Appendix~\ref{appendix:figures}). Moreover, examples exist where two distinct binary trees with $n\geq 4$ leaves have the same/different $\mathcal{B}$ and/or the same/different $\mathcal{N}$ and/or the same/different $\Delta$. A (unique) minimal example in terms of the leaf number $n$ for each possible pair of sequences is given in Figures \ref{Fig:6_5_6}, \ref{Fig:9_42_44}, \ref{Fig:11_194_199}, \ref{Fig:11_201_205} and \ref{Fig:13_869_957} in Appendix~\ref{appendix:figures}. For an overview, see also Table \ref{Tab:Uebersicht_Bsp_Seq}, which indicates the corresponding figures for each case.

\begin{table}[htbp]
    \caption{This table provides an overview of where to find a minimal example of two distinct binary trees that either share or differ in two of the three sequences $\mathcal{B}$, $\mathcal{N}$, and $\Delta$. Note that all figures except for Figure \ref{Fig:13_869_957} show a unique minimal example.}
	\label{Tab:Uebersicht_Bsp_Seq}
	\centering
	\begin{tabular}{c|c|c|c}
	\textbf{First Sequence} & \textbf{Second Sequence} & \textbf{Figure} & $\boldsymbol{n}$\\
	\hline \hline
    $\mathcal{B}(T_1) = \mathcal{B}(T_2)$ & $\mathcal{N}(T_1) = \mathcal{N}(T_2)$ & \ref{Fig:9_42_44} & $9$\\
	$\mathcal{B}(T_1) \neq \mathcal{B}(T_2)$ & $\mathcal{N}(T_1) = \mathcal{N}(T_2)$ & \ref{Fig:11_201_205} & $11$\\
	$\mathcal{B}(T_1) = \mathcal{B}(T_2)$ & $\mathcal{N}(T_1) \neq \mathcal{N}(T_2)$ & \ref{Fig:13_869_957} & $13$\\
	\hline
	$\mathcal{B}(T_1) = \mathcal{B}(T_2)$ & $\Delta(T_1) = \Delta(T_2)$ & \ref{Fig:11_194_199} & $11$\\
	$\mathcal{B}(T_1) \neq \mathcal{B}(T_2)$ & $\Delta(T_1) = \Delta(T_2)$ & \ref{Fig:6_5_6} & $6$\\
	$\mathcal{B}(T_1) = \mathcal{B}(T_2)$ & $\Delta(T_1) \neq \Delta(T_2)$ & \ref{Fig:9_42_44} & $9$\\
    \hline
	$\mathcal{N}(T_1) = \mathcal{N}(T_2)$ & $\Delta(T_1) = \Delta(T_2)$ & \ref{Fig:11_194_199} & $11$\\
    $\mathcal{N}(T_1) \neq \mathcal{N}(T_2)$ & $\Delta(T_1) = \Delta(T_2)$ & \ref{Fig:6_5_6} & $6$\\
	$\mathcal{N}(T_1) = \mathcal{N}(T_2)$ & $\Delta(T_1) \neq \Delta(T_2)$ & \ref{Fig:9_42_44} & $9$
	\end{tabular}
\end{table}

\subsection{Metaconcepts}
\label{Subsec:Metaconcepts}

In this section, we analyze the three previously introduced metaconcepts. 
Specifically, we determine which families of the function $f$ ensure that a given metaconcept yields a (binary) imbalance index. Accordingly, we examine the trees that minimize and maximize each metaconcept based on the choice of $f$ and provide formulas for computing their minimum and maximum values. Finally, we analyze the locality and the recursiveness of the metaconcepts.

Throughout this manuscript, we focus exclusively on first-order metaconcepts. However, all results regarding minimizing and maximizing trees extend to higher-order metaconcepts that are equivalent to the first-order case. For examples of such functions, see the following remark.

\begin{Rem}
\label{Rem:only_first_order}
    The BVM, the LDM, and the binary CSM of order $\omega \geq 2$ with a function of the form $f(x,o_1, \ldots, o_{\omega-1}) = f_1(x) \cdot f_2(o_1, \ldots, o_{\omega-1}) + f_3(o_{1}, \ldots, o_{\omega-1})$ are equivalent to the first-order metaconcept with $f(x) = f_1(x)$, provided that $f_2(o_1, \ldots, o_{\omega-1}) > 0$ and the additional values $o_1, \ldots, o_{\omega-1}$ are the same for all trees with the same number of leaves (e.g., $o_i = n$ but $o_i \neq h(T)$). This equivalence holds because the additional values act as constants, and the number of summands in the calculation of these metaconcepts remains the same for all trees with $n$ leaves.

    For arbitrary trees, the number of summands in the CSM varies. Thus, the CSM of order $\omega \geq 2$ can only be guaranteed to be equivalent to the first-order metaconcept if the function is of the form $f(x,o_1, \ldots, o_{\omega-1}) = f_1(x) \cdot f_2(o_1, \ldots, o_{\omega-1})$, where $f_2(o_1, \ldots, o_{\omega-1}) > 0$. In this case, the metaconcept remains equivalent to the first-order metaconcept with function $f(x) = f_1(x)$.
\end{Rem}

\paragraph*{Summary of our main results}

First, we outline the conditions on the function $f$ that ensure that the respective metaconcept yields a (binary) imbalance index. For a detailed overview, including the minimizing trees on $\BT$, see Table \ref{Tab:Uebersicht_results}. In the next step, we analyze the locality and the recursiveness of the metaconcepts.

\begin{table}[ht]
\footnotesize
\caption{This table provides an overview of our results, indicating for which families of the function $f$ the metaconcepts qualify as (binary) imbalance indices. The column labels refer to four cases, all of which require $f$ to be strictly increasing. A checkmark ($\checkmark$) indicates that no further conditions on $f$ are needed to satisfy the corresponding property. When additional constraints on $f$ are required, they are explicitly stated in the respective cell. Conversely, a cross ($\times$) indicates that for at least one function in the given family, the metaconcept fails to be a (binary) imbalance index. The entry \enquote{depends} means that the binary minimizing tree(s) are not the same for all functions within that family. An entry in square brackets indicates that this result is adapted from \citet[Corollary 4.4]{Cleary2025}.}
\label{Tab:Uebersicht_results}
\begin{tabular}{ccc|c|c|c|c|}
    \cline{4-7}
    & &  &\multicolumn{4}{|c|}{$f$ strictly increasing and} \\ \cline{4-7}

    \multicolumn{3}{c|}{} & \textbf{--} & convex & str. concave & affine ($m > 0$, $a \in \R$) \\ \hline \cline{4-7}

    \multicolumn{1}{|c|}{\multirow{2}{*}{$\Phi^{\mathcal{B}}_{f}$}} & \multicolumn{2}{c||}{imb. index on $\BT$} & \bigcell{c}{\checkmark\\Theo. \ref{Theo:bal_meta_imbalance_index}} & \bigcell{c}{\checkmark\\Theo. \ref{Theo:bal_meta_imbalance_index}} & \bigcell{c}{\checkmark\\Theo. \ref{Theo:bal_meta_imbalance_index}} & \bigcell{c}{\checkmark, $\equiv C(T)$\\Rem. \ref{Rem:bal_meta_eqiuv_Colless}} \\ \cline{2-7}

    \multicolumn{1}{|c|}{} & \multicolumn{2}{c||}{min. tree(s) on $\BT$} & \bigcell{c}{depends\\Rem. \ref{Rem:examples_result_table}} & \bigcell{c}{e.g., $T^{mb}_n$\\Theo. \ref{Theo:mb_min_bal_meta_convex}} & \bigcell{c}{depends\\Rem. \ref{Rem:examples_result_table}} & \bigcell{c}{$\arg\min C(T)$\\Cor. \ref{Cor:bal_meta_affine}} \\ \hline\hline

    \multicolumn{1}{|c|}{\multirow{3}{*}{$\Phi^{\mathcal{N}}_{f}$}} & \multicolumn{1}{c|}{\multirow{2}{*}{imb. index on}} & \multicolumn{1}{c||}{$\BT$} & \multicolumn{1}{c|}{\bigcell{c}{$\times$\\Rem. \ref{Rem:examples_result_table}}} & \multicolumn{1}{c|}{\bigcell{c}{str. convex: \checkmark\\Cor. \ref{Cor:clade_meta_convex}}} & \multicolumn{1}{c|}{\bigcell{c}{[\checkmark]\\Prop. \ref{Prop:gfb_min_clade_meta_concave_imb_index}}} & \multicolumn{1}{c|}{\bigcell{c}{\checkmark, $\equiv S(T)$\\Rem. \ref{Rem:clade_meta_eqiuv_Sackin}}} \\ \cline{3-7}

    \multicolumn{1}{|c|}{} & \multicolumn{1}{c|}{} & \multicolumn{1}{c||}{$\T$} & \multicolumn{1}{c|}{\bigcell{c}{$\times$\\Rem. \ref{Rem:examples_result_table}}} & \multicolumn{1}{c|}{\bigcell{c}{str. convex, $2$-positive: \checkmark\\Cor. \ref{Cor:clade_meta_convex}}} & \multicolumn{1}{c|}{\bigcell{c}{$2$-positive: \checkmark\\Cor. \ref{Cor:clade_meta_concave_onT}}} & \multicolumn{1}{c|}{\bigcell{c}{$a \geq 0$: \checkmark, $a = 0 \Rightarrow \equiv S(T)$\\Rem. \ref{Rem:clade_meta_eqiuv_Sackin}, Prop. \ref{Prop:Nv_imbalance_index_f_affine}}} \\ \cline{2-7}

    \multicolumn{1}{|c|}{} & \multicolumn{2}{c||}{min. tree(s) on $\BT$} & \bigcell{c}{depends\\Theo. \ref{Theo:mb_clade_meta_convex},\\Prop. \ref{Prop:gfb_min_clade_meta_concave_imb_index}} & \bigcell{c}{str. convex: $T^{mb}_n$\\Theo. \ref{Theo:mb_clade_meta_convex}} & \bigcell{c}{$\left[T^{gfb}_n\right]$\\Prop. \ref{Prop:gfb_min_clade_meta_concave_imb_index}} & \bigcell{c}{$\arg\min S(T)$\\Prop. \ref{Prop:Nv_imbalance_index_f_affine}} \\ \hline\hline

    \multicolumn{1}{|c|}{\multirow{3}{*}{$\Phi^{\Delta}_{f}$}} & \multicolumn{1}{c|}{\multirow{2}{*}{imb. index on}} & \multicolumn{1}{c||}{$\BT$} & \multicolumn{1}{c|}{\bigcell{c}{$\times$\\Rem. \ref{Rem:examples_result_table}}} & \multicolumn{1}{c|}{\bigcell{c}{\checkmark\\Prop. \ref{Prop:Delta_imbalance_index_strincr_convex}}} & \multicolumn{1}{c|}{\bigcell{c}{$\times$\\Rem. \ref{Rem:examples_result_table}}} & \multicolumn{1}{c|}{\bigcell{c}{\checkmark, $\equiv S(T)$\\Rem. \ref{Rem:leaf_meta_eqiuv_Sackin}}} \\ \cline{3-7}

    \multicolumn{1}{|c|}{} & \multicolumn{1}{c|}{} & \multicolumn{1}{c||}{$\T$} & \multicolumn{1}{c|}{\bigcell{c}{$\times$\\Rem. \ref{Rem:examples_result_table}}} & \multicolumn{1}{c|}{\bigcell{c}{\checkmark\\Prop. \ref{Prop:Delta_imbalance_index_strincr_convex}}} & \multicolumn{1}{c|}{\bigcell{c}{$\times$\\Rem. \ref{Rem:examples_result_table}}} & \multicolumn{1}{c|}{\bigcell{c}{\checkmark, $\equiv S(T)$\\Rem. \ref{Rem:leaf_meta_eqiuv_Sackin}}} \\ \cline{2-7}

    \multicolumn{1}{|c|}{} & \multicolumn{2}{c||}{min. tree(s) on $\BT$} & \bigcell{c}{depends\\Rem. \ref{Rem:examples_result_table}} & \bigcell{c}{$\arg\min S(T)$\\Prop. \ref{Prop:Delta_imbalance_index_strincr_convex}} & \bigcell{c}{depends\\Rem. \ref{Rem:examples_result_table}} & \bigcell{c}{$\arg\min S(T)$\\Prop. \ref{Prop:Delta_imbalance_index_f_affine}} \\ \hline
\end{tabular}
\end{table}

\begin{Rem}
\label{Rem:examples_result_table}
    In this remark, we provide examples to illustrate some of the cases presented in Table \ref{Tab:Uebersicht_results}.

    First, we provide an example showing that the minimizing tree of the BVM can vary for strictly increasing (and possibly strictly concave) functions $f$. To cover both cases, consider two strictly increasing and strictly concave functions, $f_1$ and $f_2$, defined as follows: $f_1(x) = \log_2\left(\frac{1}{2}x +1\right)$ and $f_2(x) = \log_2\left(\frac{3}{2}x +1\right)$.

    Let $n = 5$, where three binary trees exist: $T^{gfb}_5$, $T_1 = \left(T^{fb}_2,T^{fb}_0\right)$, and $T^{cat}_5$. For $i \in \{1,2\}$, we have
    \begin{align*}
        \Phi^{\mathcal{B}}_{f_i}\left(T^{gfb}_5\right) &= 2 \cdot f_i(0) + 2 \cdot f_i(1), \\
        \Phi^{\mathcal{B}}_{f_i}(T_1) &= 3 \cdot f_i(0) + f_i(3), \\
        \Phi^{\mathcal{B}}_{f_i}\left(T^{cat}_5\right) &= f_i(0) + f_i(1) + f_i(2) + f_i(3)
    \end{align*}
    and hence
    \begin{align*}
        \Phi^{\mathcal{B}}_{f_1}\left(T^{gfb}_5\right) &\approx 1.17, \quad \Phi^{\mathcal{B}}_{f_2}\left(T^{gfb}_5\right) \approx 2.64, \\
        \Phi^{\mathcal{B}}_{f_1}(T_1) &\approx 1.32, \quad \Phi^{\mathcal{B}}_{f_2}(T_1) \approx 2.46, \\
        \Phi^{\mathcal{B}}_{f_1}\left(T^{cat}_5\right) &\approx 2.91, \quad \Phi^{\mathcal{B}}_{f_2}\left(T^{cat}_5\right) \approx 5.78.
    \end{align*}
    Thus, for the function $f_1$, the gfb-tree is the unique minimizer, while for $f_2$, tree $T_1$ is the unique minimizer of the BVM when $n = 5$.
    \smallskip

    Second, we demonstrate that the CSM is not a (binary) imbalance index for all strictly increasing functions $f$. Specifically, we show that the fb-tree is not the unique minimizing tree on $\BT$. Consider the tree $T_2 = \left(T^{gfb}_5, T^{gfb}_3\right)$, for which the clade size sequence is \[\mathcal{N}(T_2) = (2,2,2,3,3,5,8).\] Similarly, for the fully balanced tree $T^{fb}_3$, we have \[\mathcal{N}\left(T^{fb}_3\right) = (2,2,2,2,4,4,8).\]

    Now, define the function $f_3$ as follows:
    \[f_3(x) = \begin{cases}
        x &\text{if } x \leq 3, \\
        x+2 &\text{if } x > 3.
    \end{cases}\]
    With this function, we have
    \[\Phi^{\mathcal{N}}_{f_3}(T_2) = 29 < 30 = \Phi^{\mathcal{N}}_{f_3}\left(T^{fb}_3\right).\]
    Thus, $T_2$ attains a smaller value than $T^{fb}_3$, proving that the fb-tree is not always a (unique) minimizer. Consequently, the CSM is not a (binary) imbalance index for all strictly increasing functions $f$.
    \smallskip

    Third, we show that the LDM is not a (binary) imbalance index for all strictly increasing (and possibly strictly concave) functions $f$. Again, to cover both cases, we show that the fb-tree is not the unique tree minimizing the LDM on $\BT$ for a chosen strictly increasing and strictly concave function $f_4$, defined as follows: $f_4(x) = \frac{x}{x+\frac{1}{2}}$. This function is strictly increasing. Moreover, it is strictly concave, because $f_4'(x) = \frac{1}{2\left(x^2+x+\frac{1}{4}\right)}$ is strictly decreasing for $x \in \R_{\geq0}$. Then, we have for the fully balanced tree $T^{fb}_2$,
    \[\Phi^{\Delta}_{f_4}\left(T^{fb}_2\right) = 4 \cdot f_4(2) = 3.2.\]
    For the caterpillar tree $T^{cat}_4$, we have
    \[\Phi^{\Delta}_{f_4}\left(T^{cat}_4\right) = f_4(1) + f_4(2) + 2 \cdot f_4(3) \approx 3.18\]
    Thus,
    \[\Phi^{\Delta}_{f_4}\left(T^{fb}_2\right) > \Phi^{\Delta}_{f_4}\left(T^{cat}_4\right),\]
    showing that the LDM is not a (binary) imbalance index for all strictly increasing (and possibly strictly concave) functions $f$.
    \smallskip

    Fourth, we show that the minimizing tree for the LDM with a strictly increasing and strictly concave function $f$ also depends on the choice of $f$. In the previous calculation, we observed that the caterpillar minimizes the LDM for the function $f_4$ defined above. Now, let $f_5 = \log_2$ be another strictly increasing and strictly concave function.
    Then, we have
    \[\Phi^{\Delta}_{f_5}\left(T^{fb}_2\right) = 4 < 4.17 \approx \Phi^{\Delta}_{f_5}\left(T^{cat}_4\right).\]
    Thus, $T^{fb}_2$ attains the minimum for the function $f_5$, illustrating that the minimizing tree for the LDM depends on the choice of $f$.
\end{Rem}

\paragraph*{Choosing a suitable (binary) imbalance index derived from a metaconcept}

Before measuring tree balance, three key questions must be addressed: Are the trees binary or arbitrary? Which binary tree(s) should be considered the most balanced? Which aspect of the tree (balance values, clade sizes, or leaf depths) should be used to measure balance? Once these questions are answered, the next step is to determine which imbalance index, derived from which metaconcept, is most suitable.

To support this decision, we provide four decision trees in Figures \ref{Fig:decision_tree_first_min} and \ref{Fig:decision_tree_first_meta}. Each figure contains two decision trees: one for binary trees and one for arbitrary trees. The decision trees in Figure \ref{Fig:decision_tree_first_min} begin with a choice of binary minimizing tree(s), while those in Figure \ref{Fig:decision_tree_first_meta} start with the aspect of the tree to be considered, i.e., the class of metaconcepts, and then proceed to the selection of binary minimizing trees.

In both figures, the notations \enquote{$\equiv C(T)$} and \enquote{$\equiv S(T)$} indicate that the resulting imbalance index is equivalent to the Colless or Sackin index, respectively. The label \enquote{$\argmin S(T)$} indicates that the binary minimizing trees coincide with those of the Sackin index, although the imbalance index itself may not be equivalent to the Sackin index (as it can lead to different rankings of non-extremal trees).

To illustrate the different options for binary minimizing trees, Figure \ref{Fig:min_trees_n12} shows examples of $T^{mb}_n$, $T^{gfb}_n$, $\argmin S(T)$, and, for completeness, also $\argmin C(T)$ for $n = 12$ leaves.

\begin{figure}[htbp]
    \centering
    \includegraphics[width=0.75\linewidth]{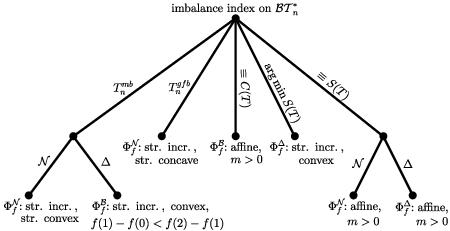}\\
    \vspace*{0.2cm}
    \includegraphics[width=0.6\linewidth]{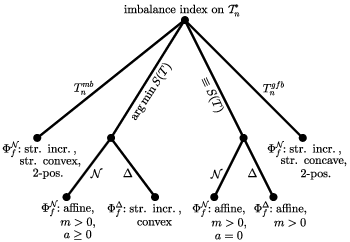}
    \caption{In these decision trees, one first selects the binary minimizing tree(s), followed by the class of metaconcepts from which the resulting (binary) imbalance index should be derived. The leaves of the trees are labeled with the corresponding metaconcept that satisfies the previously chosen properties, assuming the function $f$ is chosen as indicated.}
    \label{Fig:decision_tree_first_min}
\end{figure}

\begin{figure}[htbp]
    \centering
    \includegraphics[width=\linewidth]{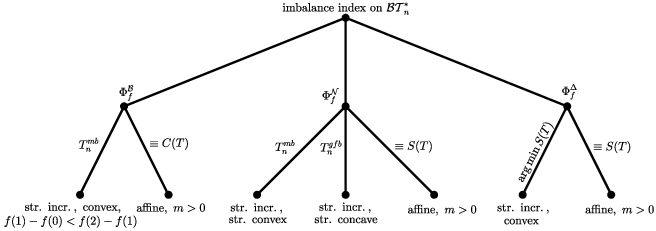}\\
    \vspace*{0.2cm}
    \includegraphics[width=0.75\linewidth]{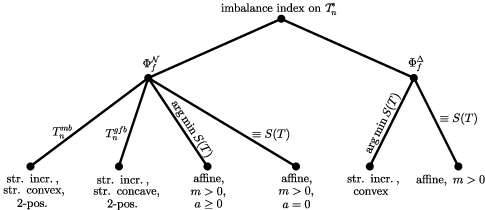}
    \caption{In these decision trees, one first selects the metaconcept from which the resulting (binary) imbalance index should be derived, and then chooses the corresponding binary minimizing tree(s). The leaves are labeled with the condition that the function $f$ must satisfy to ensure the previously chosen properties of the resulting (binary) imbalance index.}
    \label{Fig:decision_tree_first_meta}
\end{figure}

\begin{figure}[htbp]
    \centering
    \includegraphics[width=1\linewidth]{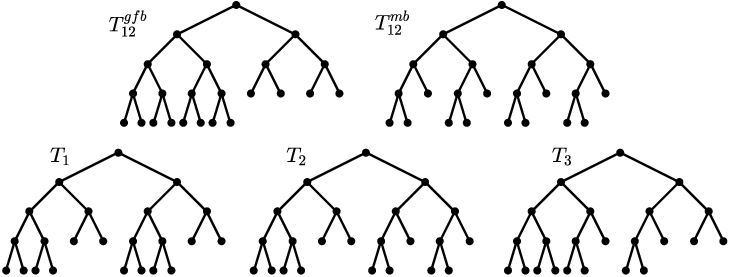}
    \caption{Depicted are all trees that minimize the Sackin index on $\mathcal{BT}^{\ast}_{12}$, including the gfb-tree and the mb-tree, i.e., $\argmin\limits_{T \in \mathcal{BT}^{\ast}_{12}} S(T) = \left\{T^{gfb}_{12}, T^{mb}_{12}, T_1, T_2, T_3\right\}$. Among these, all trees except $T_3$ also minimize the Colless index on $\mathcal{BT}^{\ast}_{12}$, i.e., $\argmin\limits_{T \in \mathcal{BT}^{\ast}_{12}} C(T) = \left\{T^{gfb}_{12}, T^{mb}_{12}, T_1, T_2\right\}$.}
    \label{Fig:min_trees_n12}
\end{figure}

\newpage 
\paragraph*{Calculating the metaconcepts in \textsf{R}}

Here, we provide \textsf{R} code to calculate the three metaconcepts using the \textsf{R} packages \texttt{ape}~\cite{Paradis2019} and \texttt{treebalance}~\cite{Fischer2023}.

%\begin{tcolorbox}[breakable,title=Calculating the metaconcepts in \textsf{R}]
%\begin{minted}{R}
\begin{codebox}
{Calculating the metaconcepts in \textsf{R}}
library(ape)
library(treebalance)

#trees in Newick format
binary_tree_newick <- "((((,),(,)),((,),(,))),((,),(,)));" #T^{gfb}_{12}
tree_newick <- "((((,),(,)),(,)),((,,),(,,)));" #tree that is not binary
#trees in phylo format
binary_tree <- ape::read.tree(text = binary_tree_newick)
tree <- ape::read.tree(text = tree_newick)


#function to calculate balance value metaconcept (BVM)
#input: binary tree in phylo format and function f (default: identity)
BVM <- function(binary_tree, f = function(b_v) {return(b_v)}) {
  if (!is.binary(binary_tree)) {
    stop("The input tree must be binary.")
  }
  n <- length(binary_tree$tip.label) #number of leaves
  #list of all clade sizes (inner vertices and leaves)
  all_clade_sizes <- get.subtreesize(binary_tree)
  Descs <- getDescMatrix(binary_tree) #determine children of vertices (matrix)
  #Rows n+1 to n+(n-1)=2n-1 in Descs represent the n-1 inner vertices.
  balance_values <- NULL #initialize list of balance values
  for (i in (n+1):(2*n-1)) { #go through all inner vertices
    #calculate the balance value of vertex i
    balance_values[length(balance_values)+1] <- abs(all_clade_sizes[Descs[i, 1]]
                                                  - all_clade_sizes[Descs[i, 2]])
  }
  return(sum(f(balance_values))) #calculate BVM
}
BVM(binary_tree)


#function to calculate clade size metaconcept (CSM)
#input: tree in phylo format and function f (default: identity)
CSM <- function(tree, f = function(n_v) {return(n_v)}) {
  #list of all clade sizes (inner vertices and leaves)
  all_clade_sizes <- get.subtreesize(tree)
  n <- length(tree$tip.label) #number of leaves
  num_vertices <- length(all_clade_sizes) #number of all vertices
  #list of clade sizes of all inner vertices
  inner_clade_sizes <- all_clade_sizes[(n+1):num_vertices]
  return(sum(f(inner_clade_sizes))) #calculate CSM
}
CSM(binary_tree)
CSM(tree)


#function to calculate leaf depth metaconcept (LDM)
#input: tree in phylo format and function f (default: identity)
LDM <- function(tree, f = function(delta) {return(delta)}) {
  n <- length(tree$tip.label) #number of leaves
  Descs <- getDescMatrix(tree) #determine children of vertices (matrix)
  #list of leaf depths (vertices 1 to n represent leaves, n+1 represents root)
  leaf_depths <- getNodesOfDepth(Descs, root = n+1, n = n)$nodeDepths[1:n]
  return(sum(f(leaf_depths))) #calculate LDM
}
LDM(binary_tree)
LDM(tree)
\end{codebox}
%\end{tcolorbox}

\paragraph*{General result regarding minimizing trees and sequences}

To investigate the extremal trees associated with these metaconcepts, we frequently use the following lemma. Note that the first part of this lemma is a generalization of \cite[Theorem 2]{Fischer2021b}.

\begin{Lem}
\label{Lem:Seq_entries_Min_Max_all_meta}
    Let $Seq$ be a sequence of length $l$, sorted in ascending order, which can be determined for every tree $T \in \mathcal{T}$, where $\mathcal{T} \subseteq \T$. Denote the $i$-th entry of $Seq(T)$ by $Seq(T)_i$. Let $f: \R \rightarrow \R$ be a function, and define the functional $\Phi^{Seq}_f: \mathcal{T} \rightarrow \R$ by
    \[\Phi^{Seq}_f(T) \coloneqq \sum\limits_{i=1}^{l} f\left(Seq(T)_i\right).\]
    Then, we have:
    \begin{enumerate}
        \item
            \begin{enumerate}
                \item If a tree $T \in \mathcal{T}$ minimizes the functional $\Phi^{Seq}_f$ on $\mathcal{T}$ for all strictly increasing functions $f$, then for all $\widetilde{T} \in \mathcal{T}$, we have
                \[Seq(T)_i \leq Seq(\widetilde{T})_i \text{ for all } i \in \left\{1,\ldots,l\right\}.\]
                \item Conversely, if a tree $T \in \mathcal{T}$ satisfies for all $\widetilde{T} \in \mathcal{T}$ and all $i \in \left\{1,\ldots,l\right\}$  
                \[Seq(T)_i \leq Seq(\widetilde{T})_i,\]
                then $T$ minimizes the functional $\Phi^{Seq}_f$ on $\mathcal{T}$ for all (not necessarily strictly) increasing functions.
            \end{enumerate}

        \item If $T$ \emph{uniquely} minimizes the functional for some increasing function $f$, then we have for all $\widetilde{T} \in \mathcal{T} \setminus \left\{T\right\}$
        \[Seq(T)_i < Seq(\widetilde{T})_i \text{ for at least one } i \in \left\{1,\ldots,l\right\}.\]
        Conversely, if
        \[Seq(T)_i \leq Seq(\widetilde{T})_i \text{ for all } i \in \left\{1,\ldots,l\right\}\]
        and
        \[Seq(T)_i < Seq(\widetilde{T})_i \text{ for at least one } i \in \left\{1,\ldots,l\right\}\]
        for all $\widetilde{T} \in \mathcal{T} \setminus \left\{T\right\}$, then $T$ (uniquely) minimizes the functional for all (strictly) increasing functions $f$.
    \end{enumerate}

    Both statements also hold in the maximization case, where \enquote{minimizing} is replaced by \enquote{maximizing}, and all inequalities are reversed.
\end{Lem}
\begin{proof}\leavevmode
    \begin{enumerate}
        \item
            \begin{enumerate}
                \item We prove this assertion by contradiction. Let $T$ minimize the functional $\Phi^{Seq}_f$ for all strictly increasing functions $f$. Assume that there exists a tree $\widetilde{T} \in \mathcal{T}$ such that $Seq(T)_i > Seq(\widetilde{T})_i$ for at least one $i \in \left\{1, \ldots, l\right\}$. For the rest of the proof let $i_{min}$ be the smallest $i$ with this property. The strategy of the proof is now to construct a function $\widetilde{f}$ that is strictly increasing and yields $\Phi^{Seq}_{\widetilde{f}}(T) > \Phi^{Seq}_{\widetilde{f}}(\widetilde{T})$, leading to a contradiction.
        
                Note that the sequence $Seq$ is sorted in ascending order. Thus, $T$ has more entries in its sequence $Seq(T)$ whose value is at least $Seq(T)_{i_{min}}$ than $\widetilde{T}$ has in its sequence $Seq(\widetilde{T})$. Let $m_{T}$ be the number of entries in $Seq(T)$ with $Seq(T)_j \geq Seq(T)_{i_{min}}$, i.e., $m_{T} = l - i_{min} + 1$. Moreover, let $m_{\widetilde{T}}$ be the number of entries in $Seq(\widetilde{T})$ with $Seq(\widetilde{T})_j \geq Seq(T)_{i_{min}}$. Then we have  $m_{\widetilde{T}} < m_{T}$ and thus $m_{\widetilde{T}} - m_{T} \leq -1$. 
                
                The idea of the construction of $\widetilde{f}$ now is to take the identity function $\id$ and add a penalty term $x$ to values greater or equal to $Seq(T)_{i_{min}}$. Let $D \coloneqq \Phi^{Seq}_{\id}(\widetilde{T}) - \Phi^{Seq}_{\id}(T)$ be the difference of the functional applied to $T$ and $\widetilde{T}$ when using the identity function $\id$, which is strictly increasing, thus implying $D\geq 0$. 

                Now, let $x\in \R_{>D}$. We then define the strictly increasing function $\widetilde{f}$ as follows: 
                \[\widetilde{f}(s) \coloneqq \begin{cases}
                                            s, \text{ if } s < Seq(T)_{i_{min}}\\
                                            s + x, \text{ if } s \geq Seq(T)_{i_{min}}\\
                                        \end{cases}.\]
                This yields
                \begin{align*}
                    \Phi^{Seq}_{\widetilde{f}}(\widetilde{T}) - \Phi^{Seq}_{\widetilde{f}}(T) &= \left(\Phi^{Seq}_{\id}(\widetilde{T}) + m_{\widetilde{T}} \cdot x\right) - \left(\Phi^{Seq}_{\id}(T) + m_{T} \cdot x\right)\\
                    &= D + \underbrace{\left(m_{\widetilde{T}} - m_{T}\right)}_{\leq -1} \cdot \underbrace{x}_{> D} < 0.
                \end{align*}
                This contradicts the assumption that $T$ minimizes the functional for all strictly increasing functions and thus completes the proof for this part.

                \item Now, let $T \in \mathcal{T}$ and suppose that for all $\widetilde{T} \in \mathcal{T}$ and for all $i\in \{1,\ldots,l\}$, we have $Seq(T)_i \leq Seq(\widetilde{T})_i.$
                For any increasing function $f$, it follows that applying $f$ to each entry preserves the order, i.e.,
                $f(Seq(T)_i) \leq f(Seq(\widetilde{T})_i)$ for all  $i \in \left\{1,\ldots,l\right\}$.
                Summing over all indices, we obtain
                \[ \Phi^{Seq}_f(T) = \sum\limits_{i=1}^l f(Seq(T)_i) \leq \sum\limits_{i=1}^l f(Seq(\widetilde{T})_i) = \Phi^{Seq}_f(\widetilde{T}).\]
                Thus, $T$ minimizes $\Phi^{Seq}_f$ for all increasing functions $f$, which completes the proof of this part.
            \end{enumerate}

        \item First, assume that $T$ uniquely minimizes the functional $\Phi^{Seq}_f$ for some increasing function $f$. We want to show that
        \[Seq(T)_i < Seq(\widetilde{T})_i \text{ for at least one } i \in \left\{1,\ldots,l\right\}
        \text{ and for all }\widetilde{T} \in \mathcal{T} \setminus \{T\}.\] Assume that this is not the case, i.e., assume there exists $\widehat{T} \in \mathcal{T} \setminus \{T\}$ such that
        \[Seq(T)_i \geq Seq(\widehat{T})_i \text{ for all } i \in \left\{1,\ldots,l\right\}.\]
        Then, while $T$ might minimize the functional, $T$ cannot minimize the functional uniquely, because, by the first part of the lemma, $\Phi^{Seq}_f(T) \geq \Phi^{Seq}_f(\widehat{T})$. This contradicts the assumption and completes the proof for this part.
        \smallskip

        For the second assertion, let $T \in \mathcal{T}$ and assume that
        \[Seq(T)_i \leq Seq(\widetilde{T})_i \text{ for all } i \in \left\{1,\ldots,l\right\}\]
        and
        \[Seq(T)_i < Seq(\widetilde{T})_i \text{ for at least one } i \in \left\{1,\ldots,l\right\}\]
        for all $\widetilde{T} \in \mathcal{T} \setminus \{T\}$. Then, applying any (strictly) increasing function $f$, we obtain
        \[ f(Seq(T)_i) \leoq f(Seq(\widetilde{T})_i) \text{ for all } i \in \left\{1,\ldots,l\right\}\]
        with strict inequality for at least one index precisely if $f$ is strictly increasing. Summing over all indices gives
        \[\Phi^{Seq}_f(T) = \sum\limits_{i=1}^l f(Seq(T)_i) \leoq \sum\limits_{i=1}^l f(Seq(\widetilde{T})_i) = \Phi^{Seq}_f(\widetilde{T}),\]
        which shows that $T$ (uniquely) minimizes the functional for all (strictly) increasing functions $f$. This completes the proof for minimization.
    \end{enumerate}

    The proof for the respective maximization statements follows analogously by reversing all inequalities. Thus, the entire proof is complete.
\end{proof}

Note that if $\mathcal{T} = \BT$ and $Seq \in \left\{\mathcal{B}, \mathcal{N}, \Delta\right\}$, then the sequences have the same length for all considered trees. Moreover, the functional $\Phi^{Seq}_f$ corresponds to the respective metaconcept.

Having established this useful lemma, we can now begin our analysis of the three classes of metaconcepts. We start with the balance value metaconcept.

\subsubsection{Balance value metaconcept \texorpdfstring{$\Phi^{\mathcal{B}}_f$}{BVM}}
\label{Sec:BVM}

In the following, we  prove one of our main results, namely that the BVM is a binary imbalance index for all strictly increasing functions $f$. This generalizes existing results in the literature, which were previously proven only for specific functions $f$, such as the Colless index, the corrected Colless index, and the quadratic Colless index. Additionally, we demonstrate that all imbalance indices induced by strictly increasing and affine functions $f$ are equivalent to the Colless index. Furthermore, we show that the mb-tree (uniquely) minimizes the BVM if $f$ is strictly increasing and (locally strictly) convex. In a second step, we compute the minimum and maximum values of the BVM. \medskip

We first investigate the relationship between the BVM $\Phi^{\mathcal{B}}_f$ and the Colless index, the corrected Colless index, and the quadratic Colless index by focusing on the specific properties of the function $f$ that induces each of these indices.

\begin{Rem}
\label{Rem:bal_meta_eqiuv_Colless}
    Let $f$ be an affine function, i.e., $f(b) = m \cdot b + a$, and $T \in \BT$. Recalling that a binary tree with $n$ leaves has $n-1$ inner vertices, we have
    \[\Phi^{\mathcal{B}}_f\left(T\right) = \sum\limits_{b \in \mathcal{B}\left(T\right)} f(b) = \sum\limits_{b \in \mathcal{B}\left(T\right)} (m \cdot b + a) = \left(m \cdot \sum\limits_{b \in \mathcal{B}\left(T\right)} b\right) + (n-1) \cdot a \stackrel{\text{cf. Table \ref{Tab:imbalance_indices_BT}}}{=} m \cdot C\left(T\right) + (n-1) \cdot a.\]
    It follows immediately that the BVM with strictly increasing and affine $f$ (i.e., $m > 0$) is equivalent to the Colless index on $\BT$.

    Furthermore, by Remark \ref{Rem:only_first_order}, the corrected Colless index is equivalent to the Colless index. Additionally, we note that the functions inducing the Colless index (i.e., the identity function) and the quadratic Colless index (i.e., $f_{QC}(b) = b^2$) are both strictly increasing. The function inducing the Colless index is affine, whereas the function inducing the quadratic Colless index is strictly convex for $b \geq 0$.
\end{Rem}

With this in mind, we now turn our attention to the extremal trees of the BVM.

\paragraph{Extremal trees}\leavevmode\\

We begin by analyzing the trees that maximize, respectively minimize, the BVM for (strictly) increasing functions $f$.

\begin{Theo}
\label{Theo:bal_meta_imbalance_index}
    Let $f$ be a (strictly) increasing function.
    \begin{enumerate}
        \item The caterpillar $T^{cat}_n$ is the (unique) tree maximizing the balance value metaconcept $\Phi^{\mathcal{B}}_f$ on $\BT$.
        \item If $n = 2^h$ with $h \geq 0$, then the fully balanced tree $T^{fb}_h$ is the (unique) tree minimizing $\Phi^{\mathcal{B}}_f$ on $\BT$.
    \end{enumerate}
    In particular, the balance value metaconcept is a binary imbalance index for all strictly increasing functions $f$.
\end{Theo}
\begin{proof}
    Let $f$ be a (strictly) increasing function.
    \begin{enumerate}
        \item We will show that the caterpillar (uniquely) maximizes $\Phi^{\mathcal{B}}_f$ on $\BT$. Specifically, we first show that for all trees $T \in \BT$, we have $\mathcal{B}(T)_i \leq \mathcal{B}\left(T^{cat}_n\right)_i$ for all $i \in \left\{1, \ldots, n-1\right\}.$ For $n \leq 3$, there is nothing to show, as there exists only one binary tree with $n$ leaves. Let $n \geq 4$ be the smallest number of leaves for which there exists a tree $T \in \BT$ such that $\mathcal{B}(T)_i > \mathcal{B}\left(T^{cat}_n\right)_i$ for at least one $i \in \left\{1, \ldots, n-1\right\}$. By assumption, the statement of the theorem holds for $T^{cat}_{n-1}$. Let $T_{n-1}$ be a tree from which $T$ can be obtained by attaching a cherry to one of its leaves. Note that $T^{cat}_n$ can be obtained in the same way from $T^{cat}_{n-1}$. Thus, for both trees, we have $\mathcal{B}(T)_{1} = \mathcal{B}\left(T^{cat}_n\right)_1 = 0$, corresponding to the parent of the attached cherry in each respective tree.

        Note that $\mathcal{B}\left(T^{cat}_n\right) = (0,1, \ldots, n-2)$ and, consequently, $\mathcal{B}(T^{cat}_{n-1}) = (0,1, \ldots, n-3)$. Moreover, attaching a cherry to $T_{n-1}$ to obtain $T$ increases each balance value by at most one. Therefore,
        \[\mathcal{B}\left(T^{cat}_n\right)_i = \mathcal{B}\left(T^{cat}_{n-1}\right)_{i-1} +1 \geq \mathcal{B}(T_{n-1})_{i-1} +1 \geq \mathcal{B}(T)_i \text{ for } i = 2, \ldots, n-1,\]
        contradicting the assumption that $\mathcal{B}\left(T^{cat}_n\right)_i < \mathcal{B}(T)_i$ for at least one $i \in \left\{1, \ldots, n-1\right\}$.

        Next, we show that $\mathcal{B}(T)_i < \mathcal{B}\left(T^{cat}_n\right)_i$ for at least one $i$. Since $T \neq T^{cat}_n$, this follows directly from the fact that $T$ has at least two cherries. Hence, $\mathcal{B}(T)_2 = 0 < 1 = \mathcal{B}\left(T^{cat}_n\right)_2$.

        Now, together with Lemma \ref{Lem:Seq_entries_Min_Max_all_meta}, this implies that the caterpillar (uniquely) maximizes $\Phi^{\mathcal{B}}_f$ on $\BT$ for any (strictly) increasing function $f$.

        \item Now, we show that the fb-tree of height $h$ (uniquely) minimizes $\Phi^{\mathcal{B}}_f$ on $\BT$ with $n = 2^h$ for any (strictly) increasing function $f$. Recall that we have $\mathcal{B}(T) = (0, \ldots, 0)$ if and only if $T = T^{fb}_h$. In particular, for all $i \in \left\{1,2, \ldots, n-1\right\}$, we have
        \[ \mathcal{B}\left(T^{fb}_h\right)_i = 0 \leq \mathcal{B}(T)_i,\] and if $T \neq T^{fb}_h$, then
        \[\mathcal{B}\left(T^{fb}_h\right)_i = 0 < \mathcal{B}(T)_i\]
        for at least one $i \in \left\{1,2, \ldots, n-1\right\}$. It now follows from Lemma \ref{Lem:Seq_entries_Min_Max_all_meta} that the fb-tree (uniquely) minimizes $\Phi^{\mathcal{B}}_f$ in this case, which completes the proof.
    \end{enumerate}
    By both parts of the proof, the BVM is a binary imbalance index for all strictly increasing functions. This completes the proof.
\end{proof}

Hence, by Theorem \ref{Theo:bal_meta_imbalance_index}, we have identified a family of binary imbalance indices, some of which are already known. As shown in Table \ref{Tab:imbalance_indices_BT} and Remark \ref{Rem:bal_meta_eqiuv_Colless}, this family includes the Colless index, the equivalent corrected Colless index, and the quadratic Colless index. 
\smallskip

Next, we consider the minimizing trees of the BVM with affine functions $f$. 

\begin{Cor}
\label{Cor:bal_meta_affine}
    Let $f$ be a strictly increasing affine function, i.e., $f(b) = m \cdot b + a$ with $m,a \in \R$ and $m > 0$. Then, for all $n$, the trees that minimize the balance value metaconcept $\Phi^{\mathcal{B}}_f$ are the same as those that minimize the Colless index. In particular, both the gfb-tree and the mb-tree achieve this minimum.
\end{Cor}
\begin{proof}
    The proof follows directly from the equivalence of the BVM to the Colless index, as stated in Remark \ref{Rem:bal_meta_eqiuv_Colless}, given that $m > 0$.
\end{proof}

We remark that all binary trees with $n$ leaves minimizing the Colless index are completely characterized (Remark~\ref{Rem:colless-minima}), implying that we also have a full characterization of the trees minimizing the BVM $\Phi^{\mathcal{B}}_f$ for strictly increasing affine functions $f$.
\smallskip

Next, we prove that the mb-tree minimizes the BVM for all strictly increasing and convex functions $f$.

\begin{Theo}
\label{Theo:mb_min_bal_meta_convex}
    Let $f$ be a strictly increasing and convex function. Then, the mb-tree $T^{mb}_n$ minimizes the balance value metaconcept $\Phi^{\mathcal{B}}_f$ on $\BT$ for all $n$. Moreover, the mb-tree uniquely minimizes $\Phi^{\mathcal{B}}_f$ on $\BT$ if additionally $f(1) - f(0) < f(2) - f(1)$, i.e., if $f$ is additionally locally strictly convex.
\end{Theo}

To prove this theorem we need three more lemmas. Recalling that $c_n$ denotes the minimum value of the Colless index on $\BT$, we first show that $c_n \geq 2$ for all $n \geq 4$ that are not powers of two.

\begin{Lem}
\label{Lem:collesGEQ2}
    Let $n \in \N_{\geq 4}$ be such that $n \neq 2^h$ for all $h \in \N$. Then, we have $c_n \geq 2$.
\end{Lem}
\begin{proof}
    Let $n$ be as stated in the lemma. Let $T$ be a tree minimizing the Colless index, i.e., $C(T) = c_n$. In particular, this implies $n \geq 5$. Seeking a contradiction, assume $c_n \leq 1$. Note that $c_n = 0$ would imply $n = 2^h$ as only $T^{fb}_h$ can obtain $c_n = 0$ (\citet[Corollary 1]{Coronado2020a}). Thus, we necessarily have $c_n = 1$. This, however, means that we have precisely one vertex $u$ in $T$ with balance value $1$, so its children, say $v$ and $w$, induce subtrees of sizes $n_v$ and $n_w$ with $n_w = n_v +1$. This implies that precisely one of the values $n_v$ and $n_w$, say $n_v$, is odd, and thus $n_u = n_v + n_w$ is odd, too. Now if $n_v > 1$, then $v$ would be an inner vertex with $b_v \geq 1$ (as $\left\lceil\frac{n_v}{2}\right\rceil > \left\lfloor\frac{n_v}{2}\right\rfloor$), a contradiction to $c_n = 1$. So we must have $n_v = 1$. However, as $b_u = 1$, this implies $n_w = 2$. So $n_u = 2+1 = 3$. Thus, $u$ cannot be the root of the tree as $n \geq 5$. So $u$ must have a parent $a$ of balance value $0$ (as $u$ is the only vertex with balance value $1$). Thus, $a$ has two children vertices, both of which induce a subtree of size $3$ -- and thus, as $3$ is odd, a vertex of balance value $1$. This contradiction completes the proof.
\end{proof}

Note that we already know for (strictly) increasing $f$ that $T^{mb}_n$ (uniquely) minimizes the BVM $\Phi_f^{\mathcal{B}}$ on $\BT$ if $n=2^h$ for some $h\in \mathbb{N}$. This is due to the fact that in this case, $T^{mb}_n$ coincides with $T_h^{fb}$ (see Theorem \ref{Theo:bal_meta_imbalance_index}). The following two lemmas addresses the case in which $n$ is not a power of two and shows that under certain conditions, $T^{mb}_n$ is then still the (unique) minimizer of $\Phi_f^{\mathcal{B}}$.

\begin{Lem}
\label{Lem:BAL_convex_Tmb}
    Let $f$ be strictly increasing. Let $n \in \N$ such that $n \neq 2^h$ for all $h \in \N$. If we have for all sequences $b_1, \ldots, b_k$ with $k = C\left(T^{mb}_n\right)$, $b_1 + \ldots + b_k \geq k$, and $b_i > 1$ for some $i \in \{1,\ldots,k\}$ that $k \cdot f(1) \leoq f(b_1) + \ldots + f(b_k)$, then $T^{mb}_n$ (uniquely) minimizes the balance value metaconcept $\Phi^{\mathcal{B}}_f$ on $\BT$.
\end{Lem}
\begin{proof}
    Let $f$ and $n$ be as stated in the lemma. Note that this implies that the smallest value of $n$ we need to consider is $n=3$ (as $n = 1 = 2^0$ and $n = 2 = 2^1$). In this case, however, there is only one possible binary tree, so there is nothing to show.
    Now, let $n \in \N_{\geq 4}$ be such that $n \neq 2^h$ for all $h \in \N$. In particular, we can assume $n \geq 5$. Let $k = C\left(T^{mb}_n\right)$, which satisfies $k \leq C(T)$ for all $T \in \BT$ by Proposition~\ref{Prop:gfb_mb_min_Colless}. We also know by Lemma \ref{Lem:collesGEQ2} that $k \geq 2$. Furthermore, the balance value sequence of the mb-tree consists of $k$ entries of $1$ and $n-1-k$ entries of 0, since all $n-1$  inner vertices are balanced. Consequently,
    \[\Phi^{\mathcal{B}}_f\left(T^{mb}_n\right) = k \cdot f(1) + (n-1 -k) \cdot f(0).\]
     Now let $T \in \BT\setminus\left\{T^{mb}_n\right\}$ and let $b_1, \ldots, b_l \geq 1$ be the $l \geq 1$ entries of $\mathcal{B}(T)$ that are positive. Then, we have:
    \[\Phi^{\mathcal{B}}_f(T) = \sum\limits_{i = 1}^{l} f(b_i) + (n-1-l) \cdot f(0).\]

    This leads to
    \begin{align}
        \Phi^{\mathcal{B}}_f(T) - \Phi^{\mathcal{B}}_f\left(T^{mb}_n\right) &= \left(\sum\limits_{i = 1}^{l} f(b_i) + (n-1-l) \cdot f(0)\right) - (k \cdot f(1) + (n-1 -k) \cdot f(0)) \nonumber\\
        &= \sum\limits_{i = 1}^{l} f(b_i) - (l-k) \cdot f(0) - k \cdot f(1) \label{Eq:fb_f0_f1}
    \end{align}

    Next, we distinguish two cases. Note that $b_1 + \ldots + b_l = C(T) \geq k$ and at least one $b_i > 1$ (otherwise all inner vertices of $T$ would be balanced, contradicting $T \neq T^{mb}_n$).
    \begin{itemize}
        \item First, let $l \geq k$. In this case, we have $l - k \geq 0$. Then,
            \[\text{Eq. }\eqref{Eq:fb_f0_f1} \stackrel{f \text{ str. incr.}}{\geq} \sum\limits_{i = 1}^{l} f(b_i) -(l-k) \cdot f(1) - k \cdot f(1) = \sum\limits_{i = 1}^{l} f(b_i) - l \cdot f(1) > 0,\]
            where the last inequality follows from the fact that $b_1, \ldots, b_l \geq 1$ with at least one $b_i > 1$ and the fact that $f$ is strictly increasing.
        \item Second, let $l < k$. Then, define $b_i \coloneqq 0$ for $i = l+1, \ldots, k$. In this case, we can conclude that
            \[\text{Eq. }\eqref{Eq:fb_f0_f1} = \left(\left(\sum\limits_{i=1}^k f(b_i)\right) - (k-l) \cdot f(0)\right) - (l-k) \cdot f(0) - k \cdot f(1) = \sum\limits_{i = 1}^{k} f(b_i) - k \cdot f(1) \geoq 0,\]
            where the last inequality follows from the properties of $f$ assumed in this lemma.
    \end{itemize}
    Thus, in both cases, we have $\Phi^{\mathcal{B}}_f(T) - \Phi^{\mathcal{B}}_f\left(T^{mb}_n\right)\geoq 0$ and thus, for strictly increasing $f$ satisfying $k \cdot f(1) \leoq f(b_1) + \ldots + f(b_k)$, the mb-tree is (strictly) more balanced than $T$. This completes the proof.
\end{proof}

Finally, we use the previous lemma to show that a certain family of functions $f$ satisfies the inequality in Lemma \ref{Lem:BAL_convex_Tmb} and thus yields the mb-tree as (unique) minimizer of the BVM $\Phi^{\mathcal{B}}_f$ with these functions $f$, too.

\begin{Lem}
\label{Lem:BAL_kf1_fk_Tmb}
    Let $f$ be a non-negative and strictly increasing function. Let $n \in \N$ such that $n \neq 2^h$ for all $h \in \N$. If we have $b \cdot f(1) \leoq f(b)$ for all $b \in \N_{\geq 2}$, then 
    $T^{mb}_n$ (uniquely) minimizes the balance value metaconcept $\Phi^{\mathcal{B}}_f$ on $\BT$.
\end{Lem}
\begin{proof}
    If $n \leq 3$, there is only one binary tree and thus there is nothing to show. Now, let $n \geq 4$. Since we have $n \neq 2^h$ for all $h \in \N$, we have $n \geq 5$. For this case, we show that the condition of Lemma \ref{Lem:BAL_convex_Tmb} holds for $f$. Let $k = C\left(T^{mb}_n\right)$ and $B = (b_1, \ldots, b_k)$ be any sequence with $b_1, \ldots, b_k \in \N$ such that $b_1 + \ldots + b_k \geq k$ and $b_i > 1$ for at least one $i \in \{1, \ldots, k\}$. Note that by Lemma \ref{Lem:collesGEQ2}, we have $k \geq 2$. Our goal is to show that $k \cdot f(1) \leoq f(b_1) + \ldots + f(b_k)$.

    Let $f$ be as stated in the lemma, in particular, assume $b \cdot f(1) \leoq f(b)$ for all $b \in \N_{\geq 2}$. Moreover, let $b_1, \ldots, b_{l_2} > 1$ be the $l_2 \geq 1$ entries of $B$ that are strictly greater than $1$. Additionally, let $l_0 \geq 0$ denote the number of entries of $B$ that are equal to $0$, and let $l_1 \geq 0$ be the number of entries equal to $1$. By assumption, $b_1 + \ldots + b_k \geq k$, and therefore, summing the non-zero values of $B$, we must have
    \[b_1 + \ldots + b_{l_2} + l_1 \geq k.\]
    Thus, we can derive the following inequality:
    \[\sum\limits_{i = 1}^{k} f(b_i) = \sum\limits_{i = 1}^{l_2} f(b_i) + l_0 \cdot \underbrace{f(0)}_{\geq 0} + l_1 \cdot f(1) \geq \sum\limits_{i = 1}^{l_2} f(b_i) + l_1 \cdot f(1) \geoq \sum\limits_{i = 1}^{l_2} b_i \cdot f(1) + l_1 \cdot f(1) \geq k \cdot f(1).\]

   Therefore, by Lemma \ref{Lem:BAL_convex_Tmb}, the mb-tree (uniquely) minimizes the BVM for all non-negative and strictly increasing functions $f$ satisfying $b \cdot f(1) \leoq f(b)$. This completes the proof.
\end{proof}

Now, we are in a position to prove Theorem \ref{Theo:mb_min_bal_meta_convex} using Lemma \ref{Lem:BAL_kf1_fk_Tmb}.

\begin{proof}[Proof of Theorem \ref{Theo:mb_min_bal_meta_convex}]
    In order to prove the theorem, first note that if $n = 2^h$ for some $h \in \N$, then $T^{mb}_n = T^{fb}_{h}$ and hence, by Theorem \ref{Theo:bal_meta_imbalance_index}, the mb-tree is the unique minimizer of the BVM $\Phi^{\mathcal{B}}_f$ for all strictly increasing functions $f$, and in particular for the function chosen in this theorem. 
    
    Hence, we can assume $n \neq 2^h$ for all $h \in \N$ and start by proving that if $f$ is strictly increasing, convex and satisfies $f(1) - f(0) < f(2) - f(1)$, then the mb-tree uniquely minimizes the BVM.  For the proof, we want to apply Lemma \ref{Lem:BAL_kf1_fk_Tmb}. Therefore, we need to show that we can assume $f$ to be non-negative. We even show that we can assume $f(0) = 0$. We can obtain this assumption from the equivalence of the BVM with function $f$ to the BVM with function $\widehat{f}(x) \coloneqq f(x) - f(0)$. Note that $\widehat{f}$ retains the properties of being strictly increasing, convex, and satisfying $\widehat{f}(1) - \widehat{f}(0) < \widehat{f}(2) - \widehat{f}(1)$. This implies that $f$ and $\widehat{f}$ have the same extremal properties, which is why we can without loss of generality assume $f \equiv \widehat{f}$ in the following; in particular, we may assume $f(0) = 0$.

    Using $f(0) = 0$, we observe that
    \[f(2) - f(1) > f(1) - f(0) = f(1).\]
    This means that the increment from $f(1)$ to $f(2)$ is greater than $f(1)$. By the convexity of $f$, it follows (recursively) that the increments from $f(x-1)$ to $f(x)$ for all $x \geq 2$ do not decrease, i.e., \[f(x) - f(x-1) \geq f(2) - f(1) > f(1).\]
    Thus, we obtain for all $b \geq 2$
    \[f(b) = f(1) + \sum\limits_{i = 2}^{b} \underbrace{f(i) - f(i-1)}_{> f(1)} > b \cdot f(1).\]

    Now, all requirements of Lemma \ref{Lem:BAL_kf1_fk_Tmb} are satisfied and we can conclude that the mb-tree uniquely minimizes the BVM for such functions $f$.

    For proving that the mb-tree (not necessarily uniquely) minimizes the BVM if $f$ only satisfies $f(1) - f(0) \leq f(2) - f(1)$ (which is satisfied for convex functions), then the strict inequalities of the proof above may no longer be strict. However, applying Lemma \ref{Lem:BAL_kf1_fk_Tmb} again, we conclude that the mb-tree still minimizes the BVM. This completes the proof.
\end{proof}

Next, based on the results of this subsection, we determine the extremal values of the balance value metaconcept.

\paragraph{Extremal values}\leavevmode\\

Building on results from the last section, we now determine the minimum and maximum values of the BVM. We first state the maximum value for all $n$ and the minimum value for $n = 2^h$ if $f$ is an increasing function. 

\begin{Prop}\label{Prop:Bal_Max_Min_val}
    Let $f$ be an increasing function. Then, the maximum value of the balance value metaconcept $\Phi^{\mathcal{B}}_f$ on $\BT$ is given by $\sum\limits_{i = 0}^{n-2} f(i)$. Furthermore, if $n = 2^h$, the minimum value of $\Phi^{\mathcal{B}}_f$ on $\BT$ is $(n-1) \cdot f(0)$.
\end{Prop}
\begin{proof}
    By Theorem \ref{Theo:bal_meta_imbalance_index}, the caterpillar attains the maximum value, while the fb-tree attains the minimum value. The formulas for the maximum and minimum values now directly follow from the facts that $\mathcal{B}(T^{cat}_n) = (0, \ldots, n-2)$ and $\mathcal{B}\left(T^{fb}_h\right) = (0, \ldots, 0)$ as well as $|\mathcal{B}(T)| = n-1$ for all binary trees with $n$ leaves.
\end{proof}

Finally, we determine the minimum value of the BVM for all $n$ when $f$ is not only strictly increasing but also either affine or convex. Recall that $c_n$ (given in Proposition \ref{Prop:Min_Colless}) denotes the minimum value of the Colless index on $\BT$.

\begin{Prop}
    For any $n \geq 1$, let $c_n$ be the minimum value of the Colless index on $\BT$.
    \begin{enumerate}
        \item If $f$ is strictly increasing and convex, the minimum value of the balance value metaconcept $\Phi^{\mathcal{B}}_f$ on $\BT$ for all $n$ is
        \[c_n \cdot f(1) + (n-1-c_n) \cdot f(0).\]
        \item If $f$ is strictly increasing and affine, i.e., $f(x) = m \cdot x+a$, the minimum value of $\Phi^{\mathcal{B}}_f$ on $\BT$ for all $n$ is
        \[m \cdot c_n + (n-1) \cdot a.\]
    \end{enumerate}
\end{Prop}
\begin{proof}\leavevmode
    \begin{enumerate}
        \item Let $f$ be a strictly increasing and convex function. From Proposition \ref{Prop:gfb_mb_min_Colless} and Theorem \ref{Theo:mb_min_bal_meta_convex}, we know that the mb-tree minimizes both the Colless index and the BVM. In particular, $C\left(T^{mb}_n\right) = c_n$.

            By definition, the balance value sequence of the mb-tree consists only of ones and zeros, it follows that there are exactly $c_n$ ones and $n-1-c_n$ zeros. Consequently, the minimum value of the BVM is given by
            \[\Phi^{\mathcal{B}}_f\left(T^{mb}_n\right) = c_n \cdot f(1) + (n-1-c_n) \cdot f(0).\]

        \item Now, let $f$ be a strictly increasing and affine function. The correctness of the formula for the minimum value of $\Phi^{\mathcal{B}}_f$ follows directly from the equivalence between the BVM and the Colless index, as stated in Remark \ref{Rem:bal_meta_eqiuv_Colless}.
    \end{enumerate}
    This completes the proof.
\end{proof}

Next, we analyze the metaconcepts generalizing the Sackin index. We begin with the clade size metaconcept.

\subsubsection{Clade size metaconcept \texorpdfstring{$\Phi^{\mathcal{N}}_f$}{CSM}}
\label{Sec:CSM}

In the following, we prove that the clade size metaconcept (CSM) is an imbalance index on $\BT$ (or $\T$) for all strictly increasing (and $2$-positive) functions $f$ that are either affine (i.e., $f(x) = m \cdot x + a$ with $m > 0$ (and $a \geq 0$)) or strictly convex. Recall that by Proposition \ref{Prop:gfb_min_clade_meta_concave_imb_index}, the CSM is an imbalance index on $\BT$ if $f$ is strictly increasing and strictly concave. Here, we will show that the CSM is an imbalance index on $\T$ if $f$ is additionally $2$-positive. In the next step, we calculate the minimum and maximum values of the CSM. \medskip

We first highlight the relationship between the CSM $\Phi^{\mathcal{N}}_f$ and several established tree imbalance indices related to clade sizes, such as the Sackin index, the average leaf depth, the $\widehat{s}$-shape statistic, and the total cophenetic index. We focus on the specific properties of the functions $f$ that induce each of these indices.

\begin{Rem}
\label{Rem:clade_meta_eqiuv_Sackin}
    Let $f$ be an affine function, i.e., $f(n_v) = m \cdot n_v + a$, and let $T \in \T$. Then, we have
    \[\Phi^{\mathcal{N}}_f(T) = \sum\limits_{n_v \in \mathcal{N}(T)} f(n_v) = \sum\limits_{n_v \in \mathcal{N}(T)} (m \cdot n_v + a) = \left(m \cdot \sum\limits_{n_v \in \mathcal{N}(T)} n_v\right) + |\mathcal{N}(T)| \cdot a = m \cdot S(T) + |\mathcal{N}(T)| \cdot a.\]
    It follows immediately that if $m > 0$, the CSM is equivalent to the Sackin index on $\BT$, since in this case, $|\mathcal{N}(T)| = n-1$ for each tree $T \in \BT$ and therefore, $|\mathcal{N}(T)| \cdot a$ is a constant. Additionally, the CSM is equivalent to the Sackin index on $\T$ if $m > 0$ and $a = 0$.

    However, if $m > 0$ and $a > 0$, the CSM is not equivalent to the Sackin index on $\T$. For example, consider the trees $T_1$ and $T_2$ shown in Figure \ref{Fig:clade_meta_notequiv_Sacking_T}. The Sackin indices are $S(T_1) = 23$ and $S(T_2) = 26$, meaning $S(T_1) < S(T_2)$. However, for the function $f(n_v) = n_v + 2$, we find that $\Phi^{\mathcal{N}}_f(T_1) = 35$ and $\Phi^{\mathcal{N}}_f(T_2) = 34$, so $\Phi^{\mathcal{N}}_f(T_1) > \Phi^{\mathcal{N}}_f(T_2)$. This confirms that the CSM and the Sackin index are not equivalent on $\T$.

    By Remark \ref{Rem:only_first_order}, we know that the average leaf depth is equivalent to the Sackin index. Additionally, the total cophenetic index is induced by the CSM, because on $\T$ it can also be expressed as
    \[\Phi(T) = \sum\limits_{v \in \mathring{V(T)} \setminus \{\rho\}} \binom{n_v}{2} = \left(\sum\limits_{v \in \mathring{V}(T)} \binom{n_v}{2}\right) - \binom{n}{2} = \sum\limits_{n_v \in \mathcal{N}(T)} \left(\binom{n_v}{2} - \frac{\binom{n}{2}}{|\mathring{V}(T)|}\right) = \Phi^{\mathcal{N}}_{f_{\Phi}}(T),\]
    where $f_{\Phi}\left(n_v,n,|\mathring{V}(T)|\right) = \binom{n_v}{2} - \frac{\binom{n}{2}}{|\mathring{V}(T)|}$. Hence, on $\T$ the total cophenetic index is induced by the third-order CSM, whereas on $\BT$ it is induced by the second-order CSM, because then $\mathring{V}(T) = n-1$.

    Thus, by Remark \ref{Rem:only_first_order}, on $\BT$ the total cophenetic index is equivalent to the first-order CSM with the function $f_{\widetilde{\Phi}}(n_v) = \binom{n_v}{2}$. Note that \citet{Knuver2024} introduced a function $\Phi^{\ast\ast}$ to measure network balance. When restricted to trees, this function coincides with $\Phi^{\mathcal{N}}_{f_{\widetilde{\Phi}}}$ (for further details, see \cite[page 95]{Knuver2024}.

    Finally, note that the functions $f$ that induce the Sackin index (i.e., the identity function), the $\widehat{s}$-shape statistic (i.e., $f_{\widehat{s}}(n_v) = \log(n_v -1)$), and $f_{\widetilde{\Phi}}$ are all strictly increasing. Moreover, the function for the Sackin index is affine, while the function for the $\widehat{s}$-shape statistic is strictly concave but not $2$-positive, since $f_{\widehat{s}}(2) = 0$. In contrast, $f_{\widetilde{\Phi}}$ is strictly convex.
\end{Rem}

\begin{figure}[htbp]
    \centering
    \includegraphics[scale=1.3]{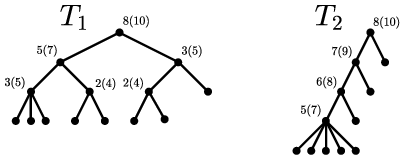}
    \caption{$T_1$ and $T_2$ are ranked differently by the Sackin index and the CSM with $f(n_v) = n_v + 2$. The inner vertices are labeled with their clade size $n_v$ followed by $f(n_v) = n_v + 2$ in brackets.}
    \label{Fig:clade_meta_notequiv_Sacking_T}
\end{figure}

With this in mind, we now turn our attention to the extremal trees of the CSM.

\paragraph{Extremal trees}\leavevmode\\

We begin our analysis of the extremal trees of the CSM by determining its maximum on $\BT$ and $\T$. By Proposition \ref{Prop:cat_max_clade_meta_str_incr}, we already know that the caterpillar uniquely attains the maximum for all strictly increasing functions $f$ on $\BT$. However, we now extend this result.

\begin{Prop}
\label{Prop:cat_max_clade_meta_BTandT}
    The caterpillar $T^{cat}_n$ (uniquely) maximizes the clade size metaconcept $\Phi^{\mathcal{N}}_{f}$ on $\BT$, provided that $f$ is a (strictly) increasing function. Moreover, the caterpillar (uniquely) maximizes the clade size metaconcept on $\T$, if $f$ is (strictly) increasing and $2$-positive, i.e., $f(x) > 0$ for $x \geq 2$.
\end{Prop}
\begin{proof}
    For the proof, we first show that for all trees $T \in \T$, the inequality
    \[\mathcal{N}^d(T)_i \leq \mathcal{N}^d\left(T^{cat}_n\right)_i \text{ for all } i \in \{1,2,\ldots,|\mathring{V}(T)|\}\]
    holds, where $\mathcal{N}^d$ denotes the clade size sequences in descending order. The main results then follow directly from this statement.

    We begin by considering the case where $T$ is not a binary tree, i.e., $T \in \T \setminus \BT$. Our goal is to transform $T$ into a binary tree while ensuring that the clade sizes of its vertices do not decrease. Since $T$ is not binary, it must contain an inner vertex with at least three children. Let $v$ be such a vertex, and denote its children by $v_1, \ldots, v_k$ with $k \geq 3$.

    We construct a new tree $T'$ from $T$ as follows: Delete the edges $(v,v_i)$ for all $2 \leq i \leq k$, introduce a new vertex $w$, and add the edges $(v,w)$ and $(w,v_i)$ for all $2 \leq i \leq k$. For an illustration, see Figure \ref{Fig:caterpillar_T_to_binary}. Note that repeating this process eventually yields a binary tree.

    Next, we compare the clade size sequences of $T$ and $T'$. Since all original vertices of $T$ remain in $T'$ with their clade sizes unchanged, the only difference is the introduction of the new vertex $w$, which contributes an additional value to the clade size sequence of $T'$. In particular, the transformation does not decrease any clade sizes.

    Thus, it suffices to establish the claim for binary trees. However, this follows directly from Proposition \ref{Prop:cat_max_clade_meta_str_incr} and Lemma \ref{Lem:Seq_entries_Min_Max_all_meta}. Hence, we have $\mathcal{N}^d(T)_i \leq \mathcal{N}^d\left(T^{cat}_n\right)_i$ for all $i \in \{1,2,\ldots,|\mathring{V}(T)|\}$ and for all $T \in \T$. In addition, by the same two results, we get $\mathcal{N}(T)_i < \mathcal{N}\left(T^{cat}_n\right)_i$ for at least one $i \in \{1,2,\ldots,n-1\}$ and for all $T \in \BT \setminus \{T^{cat}_n\}$. Thus, again by Lemma \ref{Lem:Seq_entries_Min_Max_all_meta}, the caterpillar (uniquely) maximizes $\Phi^{\mathcal{N}}_f$ on $\BT$ if $f$ is (strictly) increasing.
   
   To complete the proof, we now show that for any $T \in \T \setminus \BT$ with $T \neq T^{cat}_n$, we have $\Phi^{\mathcal{N}}_f(T^{cat}_n) > \Phi^{\mathcal{N}}_f(T)$ whenever $f$ is increasing and $2$-positive. Let $m = |\mathring{V}(T)| < n-1$. Then, we compute
    \[\Phi^{\mathcal{N}}_{f}\left(T^{cat}_n\right) = \sum\limits_{i = 1}^{m} f\left(\mathcal{N}^d\left(T^{cat}_n\right)_i\right) + \underbrace{\sum\limits_{i = m+1}^{n-1} f\left(\mathcal{N}^d\left(T^{cat}_n\right)_i\right)}_{>0, \text{ since $f$ $2$-pos.}} > \sum\limits_{i = 1}^{m} f\left(\mathcal{N}^d\left(T^{cat}_n\right)_i\right) \stackrel{f \text{ incr.}}\geq \sum\limits_{i = 1}^{m} f\left(\mathcal{N}^d\left(T\right)_i\right) = \Phi^{\mathcal{N}}_{f}\left(T\right).\]
    Thus, the caterpillar (uniquely) maximizes $\Phi^{\mathcal{N}}_f$ on $\T$ whenever $f$ is (strictly) increasing and $2$-positive, thereby completing the proof.
\end{proof}

\begin{figure}[htbp]
    \centering
    \includegraphics[scale=1.5]{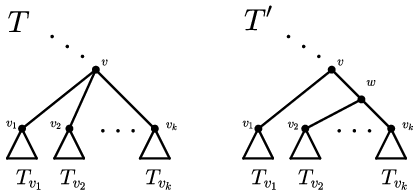}
    \caption{Trees $T$ and $T'$ as described in the first part of the proof of Proposition \ref{Prop:cat_max_clade_meta_BTandT}.}
    \label{Fig:caterpillar_T_to_binary}
\end{figure}

Next, we use this result to conclude that the CSM is an imbalance index for a certain family of functions $f$.

\begin{Cor}
\label{Cor:clade_meta_concave_onT}
    Let $f$ be a $2$-positive function, i.e., $f(x) > 0$ if $x \geq 2$, that is also strictly increasing and strictly concave. Under these conditions, the clade size metaconcept $\Phi^{\mathcal{N}}_f$ is an imbalance index.
\end{Cor}
\begin{proof}
    The proof is a direct consequence of Proposition \ref{Prop:gfb_min_clade_meta_concave_imb_index} and Proposition \ref{Prop:cat_max_clade_meta_BTandT}.
\end{proof}

By Proposition \ref{Prop:gfb_min_clade_meta_concave_imb_index}, we already know that the gfb-tree uniquely minimizes the CSM on $\BT$ if $f$ is strictly increasing and strictly concave. Next, we will show that the mb-tree uniquely minimizes the CSM on $\BT$ if $f$ is strictly increasing and strictly convex.

\begin{Theo}
\label{Theo:mb_clade_meta_convex}
    Let $f$ be a strictly increasing and strictly convex function. Then, $T^{mb}_n$ uniquely minimizes the clade size metaconcept $\Phi^{\mathcal{N}}_f$ on $\BT$.
\end{Theo}

Before we can prove this statement, we need two helpful lemmas. In proving these lemmas, we will rely on the locality property of the CSM, which states that if two trees differ only in a rooted subtree, then the difference in their CSM values is entirely determined by these subtrees. To maintain the flow of the manuscript, we postpone the formal statement and proof of this property to Proposition~\ref{Prop:locality}. Furthermore, we note that the proofs of these lemmas and the main theorem proceed analogously to the proofs presented by \citet{Mir2013} for the total cophenetic index. We include them here for completeness.

\begin{Lem}
\label{Lem:clade_Meta_convex}
    Let $T \in \mathcal{BT}^{\ast}_{n \geq 4}$ and suppose that $T$ contains a subtree $T_z$ rooted at an inner vertex $z$ with $n_T(z) \geq 4$. Suppose that $a$ and $b$ are the children of $z$ and suppose that they are both inner vertices, inducing subtrees $T_a = (T_1,T_2)$ and $T_b = (T_3,T_4)$.
    Moreover, let $n_i$ denote the number of leaves of $T_i$ with $i \in \{1,\ldots,4\}$ and assume $n_1 \geq n_2$, $n_3 \geq n_4$, and $n_1 > n_3$. If $T$ minimizes the clade size metaconcept $\Phi^{\mathcal{N}}_f$ for strictly convex $f$, then $n_4 \geq n_2$.
\end{Lem}
\begin{proof}[Proof adapted from \cite{Mir2013}, Lemma 11]
    Let $T$ minimize $\Phi^{\mathcal{N}}_f$, and assume that $f$ is strictly convex. For the sake of contradiction, suppose that $n_2 > n_4$. Construct $T'_z$ from $T_z$ by swapping the positions of $T_2$ and $T_4$, so that in $T'_z$, vertex $a$ has pending subtrees $T_1$ and $T_4$, while vertex $b$ has pending subtrees $T_3$ and $T_2$. Let $T'$ be the tree obtained from $T$ by replacing $T_z$ in $T$ with $T'_z$.

    By assumption, we have $n_1 + n_2 > n_1 + n_4$ and $n_3 + n_2 > n_3 + n_4$. Defining $\lambda \coloneqq n_2 - n_4 > 0$, we can write $n_1 + n_4 = n_1 + n_2 - \lambda$ and $n_3 + n_2 = n_3 + n_4 + \lambda$. Then,
    \begin{align*}
        \Phi^{\mathcal{N}}_f(T) - \Phi^{\mathcal{N}}_f(T') &\stackrel{\text{Prop. }\ref{Prop:locality}}{=} \Phi^{\mathcal{N}}_f(T_z) - \Phi^{\mathcal{N}}_f(T'_z) = f(n_1 + n_2) + f(n_3 + n_4) - \left(f(n_1 + n_4) + f(n_3 + n_2)\right)\\
        &= f(n_1 + n_2) - f(n_1 + n_2 - \lambda) - (f(n_3 + n_4 + \lambda) - f(n_3 + n_4))\\
        &= \left(\sum\limits_{i = 1}^{\lambda} f(n_1 + n_2 +1 -i) - f(n_1 + n_2 -i)\right)\\
        &- \left(\sum\limits_{i = 1}^{\lambda} f(n_3 + n_4 + \lambda +1 -i) - f(n_3 + n_4 + \lambda -i)\right) \stackrel{f \text{ str. convex}}{>} 0.
    \end{align*}
    The last inequality holds because the assumption $n_1 > n_3$ yields $n_1 + n_2 > n_3 + n_2 = n_3 + n_4 + \lambda$. Together with $f$ being strictly convex, the $i$-th increment in the first sum is strictly greater than the $i$-th increment in the second sum. This contradicts the minimality of $T$ and thus completes the proof.
\end{proof}

\begin{Lem}
\label{Lem:clade_Meta_increasing}
    Let $T \in \mathcal{BT}^{\ast}_{n \geq 3}$ and suppose that $T$ contains a subtree $T_z$ rooted at an inner vertex $z$ with $n_T(z) \geq 3$ and such that the children of $z$ consist of an inner vertex $a$ and a leaf $x$ of $T$.
    Further, let $T_a = (T_1,T_2)$ with $n_1 \geq n_2$. If $T$ minimizes the clade size metaconcept $\Phi^{\mathcal{N}}_f$ for a strictly increasing function $f$, then $n_1 = n_2 = 1$.
\end{Lem}
\begin{proof}[Proof adapted from \citet{Mir2013}, Lemma 12]
    Let $T$ minimize $\Phi^{\mathcal{N}}_f$, and assume that $f$ is strictly increasing. For the sake of contradiction, suppose that $n_1 > 1$. Let $T'_z$ be the tree obtained from $T_z$ by switching the positions of the leaf $x$ and the subtree $T_1$, meaning that the pending subtrees of $a$ in $T'_z$ are $T_2$ and the leaf $x$. Moreover, let $T'$ be the tree obtained from $T$ by replacing the subtree $T_z$ with $T'_z$. Then,
    \[\Phi^{\mathcal{N}}_f(T) - \Phi^{\mathcal{N}}_f(T') \stackrel{\text{Prop. }\ref{Prop:locality}}{=} \Phi^{\mathcal{N}}_f(T_z) - \Phi^{\mathcal{N}}_f(T'_z) = f(n_1 + n_2) - f(n_2 + 1) \stackrel{f \text{ str. increasing}}{>} 0.\]
    This contradicts the minimization of $T$ and completes the proof.
\end{proof}

Now, we are in a position to prove Theorem \ref{Theo:mb_clade_meta_convex}. 

\begin{proof}[Proof of Theorem \ref{Theo:mb_clade_meta_convex}, adapted from \citet{Mir2013}, Theorem 13]
    First, note that for $n = 1,2,3$, there is only one binary tree, and thus there is nothing to show. Now, assume $T \in \BT\setminus\{T^{mb}_n\}$ minimizes $\Phi^{\mathcal{N}}_f$. Let $T$ have an inner vertex $z$ that is not balanced, but whose children are balanced inner vertices or leaves. Let $a$ and $b$ be the children of $z$ with $n_a \geq n_b +2$. In particular, $a$ is an inner vertex. If $b$ is a leaf, then by Lemma \ref{Lem:clade_Meta_increasing}, we have $n_a = 2$. This contradicts the assumption that $z$ is not balanced. Hence, both $a$ and $b$ must be inner vertices. Moreover, due to the choice of $z$, $a$ and $b$ are balanced. Now, we can express the structure of $T_z$ as in Lemma \ref{Lem:clade_Meta_convex}, i.e., $T_a = (T_1,T_2)$ and $T_b = (T_3,T_4)$ with $n_1 \geq n_2$ and $n_3 \geq n_4$. Exploiting the balance of $a$ and $b$, we deduce that $n_2 \in \{n_1 - 1, n_1\}$ and $n_4 \in \{n_3 - 1, n_3\}$. Further, since $n_1 + n_2 = n_a \geq n_b + 2 = n_3 + n_4 + 2$, we can conclude that $2n_1 \geq 2n_3 + 1$, and thus $n_1 > n_3$. Therefore, it follows that $n_1 > n_3 \geq n_4$ and, by Lemma~\ref{Lem:clade_Meta_convex}, we have $n_4 \geq n_2$.
    Hence, we know that $n_2 = n_1 -1$, and thus $n_2 = n_3 = n_4$. Finally, for the balance value of $z$, we compute $b_T(z) = (n_1 + n_2) - (n_3 + n_4) = 2n_2 +1 - 2n_2 = 1$. This contradicts the assumption that $z$ is not balanced and completes the proof. Hence, $T^{mb}_n$ uniquely minimizes the CSM $\Phi^{\mathcal{N}}_f$ on $\BT$ if $f$ is thus strictly increasing and strictly convex.
\end{proof}

Using Theorem \ref{Theo:mb_clade_meta_convex}, we can identify another family of functions $f$ that induces (binary) imbalance indices. Note that, by Remark \ref{Rem:clade_meta_eqiuv_Sackin}, the total cophenetic index is equivalent to one of these functions.

\begin{Cor}
\label{Cor:clade_meta_convex}
    Let $f$ be strictly increasing and strictly convex. Then, the clade size metaconcept $\Phi^{\mathcal{N}}_f$ is an imbalance index on $\BT$. Moreover, if $f$ is also $2$-positive, i.e., $f(x) > 0$ if $x \geq 2$, then the clade size metaconcept $\Phi^{\mathcal{N}}_f$ is an imbalance index on $\T$.
\end{Cor}
\begin{proof}
    The proof is a direct consequence of Proposition \ref{Prop:cat_max_clade_meta_BTandT} and Theorem \ref{Theo:mb_clade_meta_convex}.
\end{proof}

In Remark \ref{Rem:clade_meta_eqiuv_Sackin}, we observed that the CSM with an affine function $f(n_v) = m \cdot n_v + a$ is equivalent to the Sackin index on $\BT$ if $m > 0$, and on $\T$ if $m > 0$ and $a = 0$. In both cases, it follows directly that the CSM is a (binary) imbalance index. Further, we show in the following that the CSM is an imbalance index on $\T$ if $m > 0$ and $a \geq 0$ (rather than only if $a = 0$).

\begin{Prop}
\label{Prop:Nv_imbalance_index_f_affine}
    Let $f$ be an affine function, i.e., $f(n_v) = m \cdot n_v + a$ with $m,a \in \R$. Then, the clade size metaconcept $\Phi^{\mathcal{N}}_f$ is an imbalance index on $\T$ if $m > 0$ and $a \geq 0$. Moreover, the clade size metaconcept $\Phi^{\mathcal{N}}_f$ is an imbalance index on $\BT$ if $m > 0$.

    Furthermore, for all $n$, the minimizing trees on $\BT$ of the clade size metaconcept, if $m > 0$, coincide with those of the Sackin index. Specifically, these are either $T^{fb}_{h_n}$ or trees that employ precisely two leaf depths, namely $h_n-1$ and $h_n$, where $h_n = \lceil\log_2(n)\rceil$. In particular, both the gfb-tree and the mb-tree minimize the clade size metaconcept for all $n$.
\end{Prop}
\begin{proof}
    The part regarding $\BT$ and the minimizing trees on $\BT$ for $n \neq 2^{h_n}$ follows directly from Remark \ref{Rem:clade_meta_eqiuv_Sackin}, Lemma \ref{Lem:Sackin_mintree_min_value}, and Remark \ref{Rem:gfb_mb_from_fb_min_Sackin}, since in this case the CSM is equivalent to the Sackin index.

    To prove the rest of this proposition, we must show that the caterpillar is the unique tree maximizing the CSM on $\T$, and the fully balanced tree is the unique tree minimizing the CSM on $\BT$ if $n = 2^{h_n}$.
    The latter again follows from the equivalence to the Sackin index. Finally, we show that the caterpillar is the unique tree maximizing the CSM on $\T$. In Remark \ref{Rem:clade_meta_eqiuv_Sackin}, we saw that the CSM is an affine function of the Sackin index if $f(n_v) = m \cdot n_v + a$ with $m,a \in \R$ for all $T \in \T$, i.e.,
    \[\Phi^{\mathcal{N}}_f(T) = m \cdot S(T) + |\mathcal{N}(T)| \cdot a.\]
    Note that $|\mathcal{N}(T)| = |\mathring{V}(T)|$. Now, assume $m > 0$ and $a \geq 0$. We can exploit the fact that the Sackin index is an imbalance index. Let $T \in \T$ be an arbitrary tree with $T \neq T^{cat}_n$. Then,
    \[\Phi^{\mathcal{N}}_f\left(T^{cat}_n\right) = m \cdot S\left(T^{cat}_n\right) + \underbrace{\left|\mathcal{N}\left(T^{cat}_n\right)\right|}_{= n-1} \cdot a > m \cdot S(T) + \underbrace{\left|\mathcal{N}(T)\right|}_{\leq n-1} \cdot a = \Phi^{\mathcal{N}}_f(T).\]
    Here, the strict inequality follows from the fact that the Sackin index is an imbalance index, and that $m > 0$ and $a \geq 0$.
    Thus, the caterpillar maximizes the CSM on $\T$ if $m > 0$ and $a \geq 0$. This completes the proof.
\end{proof}

So far, we have only considered the minimization of the CSM on $\BT$. In the next proposition, we extend our analysis to arbitrary trees.

\begin{Prop}
\label{Prop:clade_size_meta_star_min}
    Let $T^{star}_n$ be the star tree on $n$ leaves, and let $f$ be a $2$-positive, i.e., $f(x) > 0$ if $x \geq 2$, (though not necessarily increasing) function. Then, the star tree is the unique tree that minimizes the clade size metaconcept $\Phi^{\mathcal{N}}_f$ on $\T$.
\end{Prop}
\begin{proof}
    Let $T^{star}_n$ be the star tree on $n$ leaves, and let $f$ be a $2$-positive function. For $n \leq 2$, there is only one tree, so there is nothing to show. Now, let $n \geq 3$. Consider a tree $T \in \T \setminus \left\{T^{star}_n\right\}$ distinct from the star tree. Then, $|\mathring{V}(T)| \geq 2$. Moreover, every tree has at least one inner vertex with clade size $n$, namely its root. Thus,
    \[\Phi^{\mathcal{N}}_f\left(T^{star}_n\right) = f(n) \stackrel{f\text{ $2$-pos.}}{<} f(n) + \sum\limits_{i = 1}^{|\mathring{V}(T)|-1} f(\mathcal{N}(T)_i) = \Phi^{\mathcal{N}}_f(T).\]
    This completes the proof.
\end{proof}

Having identified the trees that minimize and maximize the CSM, we can now calculate its minimum and maximum values.

\paragraph{Extremal values}\leavevmode\\

We begin by considering the maximum value of the CSM if $f$ is increasing (and $2$-positive).

\begin{Prop}
\label{Prop:Clade_size_Meta_Max_val}
    The maximum value of the clade size metaconcept $\Phi^{\mathcal{N}}_f$ on $\BT$ (on $\T$) is $\sum\limits_{i=2}^{n} f(i)$, if $f$ is an increasing (and $2$-positive, i.e., $f(x) > 0$ if $x \geq 2$) function.
\end{Prop}
\begin{proof}
    By Proposition \ref{Prop:cat_max_clade_meta_BTandT}, it suffices to show that $\Phi^{\mathcal{N}}_f\left(T^{cat}_n\right) = \sum\limits_{i=2}^{n} f(i)$. However, this follows directly from the fact that $\mathcal{N}\left(T^{cat}_n\right) = (2, \ldots, n)$. This completes the proof.
\end{proof}

Next, we consider the minimum values of the CSM. Note that \citet{Cleary2025} have stated a minimum value for their function $\pi_c$, which is highly related to the CSM on $\BT$ with function $f(n_v) = \log(n_v + c)$ with $c > -2$. Taking the logarithm of the minimum value of $\pi_c$ (\citet[Corollary 4.13]{Cleary2025}) yields the minimum value for the CSM on $\BT$ for those functions $f$. Next, we extend this result to a broader range of functions $f$.

\begin{Prop} Let $h_i = \lceil\log(i)\rceil$ for all $i \in \mathbb{N}$. 
    Then, the minimum value of the clade size metaconcept $\Phi^{\mathcal{N}}_f$ on $\BT$ for any strictly increasing and strictly concave function $f$ is $\sum\limits_{i=2}^{n} gfb_n(i) \cdot f(i)$, where $gfb_n(i)$ is as specified in Theorem~\ref{theo:gfb-subtree-sizes}.
\end{Prop}
\begin{proof}
    First, recall that $gfb_n(i)$ is the number of subtrees of $T^{gfb}_n$ of size $i$. Now, for strictly increasing and strictly concave $f$, the statement is a direct consequence of Proposition \ref{Prop:gfb_min_clade_meta_concave_imb_index}.
\end{proof}

Next, we consider the minimum value of the CSM for strictly increasing and strictly convex $f$.

\begin{Prop}
\label{Prop:min_val_CSM_convex}
    The minimum value of the clade size metaconcept $\Phi^{\mathcal{N}}_f$ for any strictly increasing and strictly convex function $f$ is $\sum\limits_{i=2}^n mb_n(i) \cdot f(i)$ with $mb_n(i)$ as given in Lemma \ref{Lem:number_subtrees_size_mb}.
\end{Prop}

Before we can prove this statement, we give the number of subtrees of a given size for the mb-tree. The proof of the following lemma can be found in Appendix~\ref{appendix:proof-lemmas}. Note that the proof of Proposition~\ref{Prop:min_val_CSM_convex} is merely a direct consequence of Theorem \ref{Theo:mb_clade_meta_convex} and the following Lemma \ref{Lem:number_subtrees_size_mb}.

\begin{Lem}
\label{Lem:number_subtrees_size_mb}
    Let $n,i \in \mathbb{N}$ such that $i\leq n$. Let $h_n=\lceil \log_2(n) \rceil$ and $h_i=\lceil \log_2(i) \rceil$. For all $l \in \{0,\ldots,h_n\}$, let $r_l^n=n - 2^l \cdot \left\lfloor\frac{n}{2^l}\right\rfloor $. Moreover, let $mb_n(i)$ denote the number of pending subtrees with $i$ leaves of $T^{mb}_n$. Then, we have:

    $mb_n(i) = \begin{cases}
        n & \mbox{if $i = 1$},\\
        2^{h_n-1} & \mbox{if $i = 2$ and $n = 2^{h_n}$},\\
        r^n_{h_n-1} & \mbox{if $i = 2$ and $n < 2^{h_n}$ and $\left\lfloor\frac{n}{2^{h_n-2}}\right\rfloor \neq \left\lceil\frac{n}{2^{h_n-1}}\right\rceil$},\\
        2^{h_n-2} - r_{h_n-2}^n + r_{h_n-1}^n& \mbox{if $i = 2$ and $n < 2^{h_n}$ and $\left\lfloor\frac{n}{2^{h_n-2}}\right\rfloor = \left\lceil\frac{n}{2^{h_n-1}}\right\rceil$},\\

        0 & \mbox{if $i\geq 3$ and $i \not\in \left\{\left\lfloor\frac{n}{2^l}\right\rfloor, \left\lceil\frac{n}{2^l}\right\rceil\right\}$ for all $l \in \left\{\left\lfloor\log_2\left(\frac{n}{i}\right)\right\rfloor,\left\lceil\log_2\left(\frac{n}{i}\right)\right\rceil\right\}$},\\
        r_l^n & \mbox{if $i \geq 3$ and $i = \left\lceil\frac{n}{2^l}\right\rceil$ with $l = \left\lceil\log_2\left(\frac{n}{i}\right)\right\rceil$ and $\left\lceil\frac{n}{2^l}\right\rceil > \left\lfloor\frac{n}{2^l}\right\rfloor$},\\
        2^l - r_l^n & \mbox{if $i \geq 3$ and $i = \left\lfloor\frac{n}{2^l}\right\rfloor$ with $l = \left\lfloor\log_2\left(\frac{n}{i}\right)\right\rfloor$}.
    \end{cases}$
\end{Lem}

In the next proposition, we consider the maximum and minimum value of the CSM if $f$ is a strictly increasing affine function (with non-negative intercept).

\begin{Prop}
\label{Prop:clade_size_Meta_Max_Min_val_affine}
    Let $f$ be an affine function, i.e., $f(n_v) = m \cdot n_v + a$ with $m,a \in \R$. Then, we have: 
    \begin{enumerate}
        \item The maximum value of the clade size metaconcept $\Phi^{\mathcal{N}}_f$ is $m \cdot \left(\frac{n \cdot (n+1)}{2} -1\right) + (n-1) \cdot a$
            \begin{enumerate}
                \item on $\T$ if $m > 0$ and $a \geq 0$ and
                \item on $\BT$ if $m > 0$.
            \end{enumerate}

        \item Let $h_n = \lceil\log_2(n)\rceil$. Then, the minimum value on $\BT$ is $m \cdot \left(-2^{h_n} + n \cdot (h+1)\right) + (n-1) \cdot a$ if $m > 0$, which equals $m \cdot h_n \cdot 2^{h_n} + (n-1) \cdot a$ if $n = 2^{h_n}$.
    \end{enumerate}
\end{Prop}
\begin{proof}
    First, consider the maximum value on $\T$ for $m > 0$ and $a \geq 0$. 
    By Proposition \ref{Prop:Nv_imbalance_index_f_affine}, we only need to show that the caterpillar attains the stated maximum value.

    By Remark \ref{Rem:clade_meta_eqiuv_Sackin}, we know that for all $T \in \T$, the CSM satisfies $\Phi^{\mathcal{N}}_f(T) = m \cdot S(T) + |\mathcal{N}(T)| \cdot a$. Together with Lemma \ref{Lem:Sackin_cat}, we then have
    \[\Phi^{\mathcal{N}}_f\left(T^{cat}_n\right) = m \cdot S\left(T^{cat}_n\right) + |\mathcal{N}\left(T^{cat}_n\right)| \cdot a = m \cdot \left(\frac{n \cdot (n+1)}{2} -1\right) + (n-1) \cdot a.\]
    This completes 1 (a).

    The remainder of the proof follows directly from the equivalence of the CSM and the Sackin index, as stated in Remark \ref{Rem:clade_meta_eqiuv_Sackin}, along with either Lemma \ref{Lem:Sackin_cat} for the maximum value or Lemma \ref{Lem:Sackin_mintree_min_value} for the minimum value, thereby completing the proof.
\end{proof}

Now, we state the minimum value on $\T$ for $2$-positive $f$.

\begin{Lem}
\label{Lem:clade_size_Meta_Min_val_T}
    Let $f$ be a $2$-positive function, i.e., $f(x) > 0$ if $x \geq 2$. Then, if $n=1$, the minimum value of the clade size metaconcept $\Phi^{\mathcal{N}}_f$ on $\T$ is $0$. Otherwise, if $n\geq 2$, the minimum value of the clade size metaconcept $\Phi^{\mathcal{N}}_f$ on $\T$ is $f(n)$.
\end{Lem}
\begin{proof}
    Let $f$ be a $2$-positive function. If $n = 1$, then there exists only one tree with no inner vertex. Hence, the CSM is the empty sum, which is $0$. Now, let $n \geq 2$. By Proposition \ref{Prop:clade_size_meta_star_min}, it suffices to show that $\Phi^{\mathcal{N}}_f\left(T^{star}_n\right) = f(n)$. This holds because the only inner vertex of $T^{star}_n$ is the root, which has a clade size of $n$. This completes the proof.
\end{proof}

Finally, we analyze the last metaconcept, the leaf depth metaconcept. As shown in Table \ref{Tab:imbalance_indices_T}, it also is a generalization of the Sackin index.

\subsubsection{Leaf depth metaconcept \texorpdfstring{$\Phi^{\Delta}_f$}{LDM}}
\label{Sec:LDM}

In the following, we will show that the leaf depth metaconcept (LDM) is a (binary) imbalance index for strictly increasing affine functions, as well as for strictly increasing and convex functions. In a second step, we will calculate the minimum and maximum values of the LDM.

Furthermore, we have already observed in Table \ref{Tab:imbalance_indices_T} that the Sackin index and the average leaf depth are induced by the first or second-order LDM, respectively. Now, we show that the LDM is an affine function of the Sackin index for all affine functions $f$ (i.e., $f(x) = m \cdot x + a$). Additionally, it is equivalent to the Sackin index on $\T$ if $f$ is strictly increasing and affine, i.e., when $m > 0$.

\begin{Rem}
\label{Rem:leaf_meta_eqiuv_Sackin}
    Let $f$ be an affine function, i.e., $f(\delta) = m \cdot \delta + a$, and let $T \in \T$. Then, we have
    \[\Phi^{\Delta}_f(T) = \sum\limits_{\delta \in \Delta(T)} f(\delta) = \sum\limits_{\delta \in \Delta(T)} (m \cdot \delta + a) = \left(m \cdot \sum\limits_{\delta \in \Delta(T)} \delta\right) + n \cdot a = m \cdot S(T) + n \cdot a\]
    This establishes the equivalence of the LDM with strictly increasing and affine $f$, i.e., $m > 0$, to the Sackin index on $\T$.
\end{Rem}

With this in mind, we now focus on the extremal trees for the leaf depth metaconcept (LDM).

\paragraph{Extremal trees}\leavevmode\\

In this section, we analyze the maximizing and minimizing trees of the LDM. Through this analysis, we identify two families of functions for which the LDM is a (binary) imbalance index: strictly increasing and convex functions, as well as strictly increasing affine functions.

\begin{Prop}
\label{Prop:Delta_imbalance_index_strincr_convex}
    Let $f$ be a strictly increasing and convex function. Then, the leaf depth metaconcept $\Phi^{\Delta}_f$ is a (binary) imbalance index. Further, for all $n$, the minimizing trees on $\BT$ of the leaf depth metaconcept coincide with those minimizing the Sackin index. These trees are either $T^{fb}_{h_n}$ or those that employ precisely two leaf depths, namely $h_n-1$ and $h_n$, where $h_n = \lceil\log_2(n)\rceil$. In particular, the gfb-tree and the mb-tree minimize the leaf depth metaconcept for all $n$.
\end{Prop}
\begin{proof}
    Let $f$ be strictly increasing and convex. To prove that the LDM is a (binary) imbalance index, we need to show that the caterpillar uniquely maximizes the metaconcept on $\T$, and the fb-tree uniquely minimizes it on $\BT$ for $n = 2^{h_n}$. For $n \leq 2$, there is nothing to show. Now, let $n \geq 3$.

    First, we proceed as we did in the first part of the proof of Proposition \ref{Prop:cat_max_clade_meta_BTandT}. Specifically, we can turn a tree $T \in \T \setminus \BT$ step by step into a binary tree. For the exact procedure and notation, refer to Figure \ref{Fig:caterpillar_T_to_binary}. By comparing the leaf depths of $T$ and $T'$, we observe the following: all leaves that are descendants of $v_1$ have the same depth in both $T$ and $T'$. For all leaves $x$ that are descendants of $v_i$ with $i \geq 2$, it holds that $\delta_{T}(x) +1 = \delta_{T'}(x)$. Based on these two observations and the fact that $f$ is strictly increasing, we conclude that $\Phi^{\Delta}_f(T) < \Phi^{\Delta}_f(T')$. Thus, it remains to show that the caterpillar is the unique maximizer of the LDM on $\BT$. Therefore, we consider a second procedure for constructing a tree from another tree, which can be viewed as the relocation of a cherry. In this case, both trees involved are binary.

    Let $T \in \BT$ be a binary tree with an inner vertex $u$ such that $u$ is the parent of the cherry formed by the two leaves $x$ and $y$, and there exists a third leaf $z$ with $\delta(u) < \delta(z)$. Let $T'$ be the tree obtained from $T$ by making $z$ the new parent of the cherry formed by $x$ and $y$. For an illustration, see Figure \ref{Fig:cherry_relocation}. As a result, $z$ becomes an inner vertex in $T'$, and $u$ becomes a leaf in $T'$. Note that $u$ and $z$, respectively, have the same depth in $T$ and $T'$. Further, we have
    \begin{align}
        \delta(u) = \delta_{T}(x) -1 \text{ and } \delta_T(x) + 1 \leq \delta_{T'}(x) = \delta(z) + 1. \quad \label{Eq:deltas_T_T'}
    \end{align}
    Note that the caterpillar can be constructed in this manner from any other binary tree.

    Now, we show that $\Phi^{\Delta}_f(T) < \Phi^{\Delta}_f(T')$. By construction, we only need to consider the vertices $x, y, z$, and $u$, as all other leaves remain unchanged in depth between $T$ and $T'$. Hence, exploiting \eqref{Eq:deltas_T_T'}, we have
    \begin{align*}
        \Phi^{\Delta}_f(T) - \Phi^{\Delta}_f(T') &= f(\delta_{T}(x)) - f(\delta_{T'}(u)) + f(\delta_{T}(y)) - f(\delta_{T'}(x)) + f(\delta_{T}(z)) - f(\delta_{T'}(y))\\
        &= f(\delta_{T}(x)) - f(\delta_{T}(x)-1) + f(\delta_{T}(x)) - f(\delta_{T'}(x)) + f(\delta(z)) - f(\delta_{T'}(x))\\
        &\stackrel{f \text{ str. incr.}}{\leq} \underbrace{f(\delta_{T}(x)) - f(\delta_{T}(x)-1) + f(\delta_{T}(x)) - f(\delta_{T}(x)+1)}_{\leq 0 \text{, $f$ convex}} + f(\delta(z)) - f(\delta(z)+1))\\
        &\stackrel{f \text{ str. incr.}}{<} 0.
    \end{align*}
    Thus, $\Phi^{\Delta}_f(T) < \Phi^{\Delta}_f(T')$. Since the caterpillar can be obtained from any other binary tree through repeated applications of this transformation, this completes the proof of this part.\\

    Next, we show that the fb-tree is the unique minimizer of $\Phi^{\Delta}_f$ on $\BT$ for $n = 2^{h_n}$. To do so, we consider the reverse operation of the cherry relocation described earlier. Let $T' \in \BT$ be a binary tree with an inner vertex $z$ such that $z$ is the parent of the cherry formed by the two leaves $x$ and $y$, and suppose there exists a third leaf $u$ with $\delta_T(z) > \delta_T(u)$. Now, let $T$ be the tree obtained from $T'$ by making $u$ the new parent of the cherry formed by $x$ and $y$. For an example, see Figure \ref{Fig:cherry_relocation}.

    In this transformation, $u$ becomes an inner vertex in $T$, while $z$ is converted into a leaf. As in the previous case, only the vertices $x, y, z$, and $u$ are affected, and the relationships given by \eqref{Eq:deltas_T_T'} still hold. Consequently, applying the same calculation as before, we obtain $\Phi^{\Delta}_f(T) < \Phi^{\Delta}_f(T')$. Since the fb-tree can be constructed step by step using this cherry relocation procedure, this establishes the property that it uniquely minimizes $\Phi^{\Delta}_f$, thus completing the second part of the proof.\\

    It remains to show that if $n \neq 2^{h_n}$, the minimizing trees on $\BT$ of the LDM also coincide with those of the Sackin index. By Lemma \ref{Lem:Sackin_mintree_min_value}, the Sackin minimizing trees are precisely the trees that employ exactly two leaf depths. Now, observe that all such trees share the same leaf depth sequence, meaning they are assigned the same value by the LDM. Furthermore, as demonstrated earlier, these trees can be systematically constructed using the second cherry relocation operation introduced in this proof. Taken together, these observations complete the proof.
\end{proof}

\begin{figure}[htbp]
    \centering
    \includegraphics[scale=1.5]{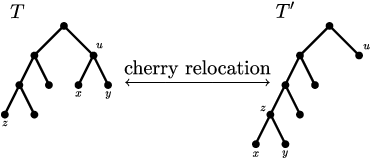}
    \caption{An example illustrating the two ways of relocating a cherry, as used in the proof of Proposition \ref{Prop:Delta_imbalance_index_strincr_convex}.}
    \label{Fig:cherry_relocation}
\end{figure}

Now, we analyze the behavior of the LDM when $f$ is a strictly increasing affine function.

\begin{Prop}
\label{Prop:Delta_imbalance_index_f_affine}
    Let $f$ be an affine function, i.e., $f(\delta) = m \cdot \delta + a$ with $m,a \in \R$. If $m > 0$, then the leaf depth metaconcept $\Phi^{\Delta}_f$ is a (binary) imbalance index. Further, for all $n$, the minimizing trees on $\BT$ of the leaf depth metaconcept if $m > 0$, coincide with those that minimize the Sackin index. Specifically, these are either $T^{fb}_{h_n}$ or trees that employ precisely two leaf depths, namely $h_n-1$ and $h_n$, where $h_n = \lceil\log_2(n)\rceil$. In particular, both the gfb-tree and the mb-tree minimize the leaf depth metaconcept for all $n$.
\end{Prop}
\begin{proof}
    The fact that $\Phi^{\Delta}_f$ is a (binary) imbalance index follows directly from its equivalence to the Sackin index, as stated in Remark \ref{Rem:leaf_meta_eqiuv_Sackin}. For the minimizing trees, in addition to the equivalence to the Sackin index, we use Lemma \ref{Lem:Sackin_mintree_min_value} and Remark \ref{Rem:gfb_mb_from_fb_min_Sackin} to establish the claim.
\end{proof}

In the next Proposition, we establish that the star tree (uniquely) minimizes the LDM on $\T$ if $f$ is (strictly) increasing.

\begin{Prop}
\label{Prop:leaf_depth_meta_star_min}
    Let $T^{star}_n$ be the star tree on $n$ leaves, and let $f$ be a (strictly) increasing function. Then, $T^{star}_n$ is the (unique) tree minimizing the leaf depth metaconcept $\Phi^{\Delta}_f$ on $\T$.
\end{Prop}
\begin{proof}
    Let $T^{star}_n$ be the star tree on $n$ leaves, and let $f$ be a (strictly) increasing function. For $n \leq 2$, there is only one tree, so there is nothing to show. Now, let $n \geq 3$, and let $T \in \T \setminus \left\{T^{star}_n\right\}$ be another tree on $n$ leaves. Since $T$ must have at least one leaf of depth strictly greater than 1, we have
    \[\Delta\left(T^{star}_n\right)_i = 1 \leq \Delta(T)_i \text{ for all } i = 1, \ldots, n\]
    and
    \[\Delta\left(T^{star}_n\right)_i = 1 < \Delta(T)_i \text{ for at least one } i = 1, \ldots, n.\]
    Since $f$ is (strictly) increasing, it follows that
    \[\Phi^{\Delta}_f\left(T^{star}_n\right) = n \cdot f(1) \leoq \sum\limits_{i = 1}^{n} f\left(\Delta(T)_i\right) = \Phi^{\Delta}_f(T).\]
    This completes the proof.
\end{proof}

Having proven these three results, we are now in a position to calculate the corresponding maximum and minimum values of the LDM.

\paragraph{Extremal values}\leavevmode\\

\begin{Prop}\leavevmode
\label{Prop:leaf_depth_Meta_Min_val_affine}
    \begin{enumerate}
        \item Let $f$ be an affine function, i.e., $f(\delta) = m \cdot \delta + a$ with $m,a \in \R$ and $m > 0$. The maximum value of the leaf depth metaconcept $\Phi^{\Delta}_f$ on both $\T$ and $\BT$ is
        \[ m \cdot \left(\frac{n \cdot (n+1)}{2} -1\right) + n \cdot a.\]
        Moreover, let $h_n = \lceil\log_2(n)\rceil$. Then, the minimum value on $\BT$ is \[m \cdot \left(-2^{h_n} + n \cdot (h_n+1)\right) + n \cdot a\]if $m > 0$, which simplifies to $m \cdot h_n \cdot 2^{h_n} + n \cdot a$ if $n = 2^{h_n}$.
        \item Let $n \geq 2$, and let $f$ be a strictly increasing and convex function. The maximum value of the leaf depth metaconcept on both $\T$ and $\BT$ is \[f(n-1) + \sum\limits_{i = 1}^{n-1} f(i).\] Moreover, let $n = 2^{h_n-1} + p$ such that $1 \leq p \leq 2^{h_n-1}$. Then, the minimum value of the leaf depth metaconcept on $\BT$ is \[\left(2^{h_n-1} - p\right) \cdot f(h_n-1) + 2 \cdot p \cdot f(h_n).\]
        \item Let $f$ be any increasing function. Then, the minimum value of the leaf depth metaconcept on $\T$ is $n \cdot f(1)$.
    \end{enumerate}
\end{Prop}

\begin{proof}\leavevmode
    \begin{enumerate}
        \item Let $f$ be an affine function, i.e., $f(\delta) = m \cdot \delta + a$ with $m,a \in \R$ and $m > 0$. The correctness of the maximum value on both $\T$ and $\BT$, as well as of the minimum value on $\BT$, follows directly from the equivalence to the Sackin index, as stated in Remark \ref{Rem:leaf_meta_eqiuv_Sackin}, and Lemma \ref{Lem:Sackin_cat} for the maximum value, respectively Lemma \ref{Lem:Sackin_mintree_min_value} for the minimum value.

        \item Let $f$ be strictly increasing and convex, and let $n \geq 2$.

        First, for the maximum value, by Proposition \ref{Prop:Delta_imbalance_index_strincr_convex}, we need to prove that $\Phi^{\Delta}_f(T^{cat}_n) = f(n-1) + \sum\limits_{i = 1}^{n-1} f(i)$. Note that the two leaves in the unique cherry in $T^{cat}_n$ have depth $n-1$, and for every smaller depth, there is exactly one leaf of this depth. Hence, the stated maximum value is correct.

        Second, for the minimum value, again by Proposition \ref{Prop:Delta_imbalance_index_strincr_convex}, we must show that the gfb-tree attains the stated minimum value. Let $n = 2^{h_n-1} + p$, where $1 \leq p \leq 2^{h_n-1}$. By Remark \ref{Rem:gfb_mb_from_fb_min_Sackin}, we know that the gfb-tree can be constructed from the fb-tree of height $h_n-1$ by attaching $p$ cherries from left to right to its leaves of depth $h_n-1$. The fb-tree of height $h_n-1$ has $2^{h_n-1}$ leaves of depth $h_n-1$ and $0$ cherries of depth $h_n$. For each attached cherry, one leaf of depth $h_n-1$ is replaced by two leaves of depth $h_n$. Therefore, after attaching $p$ cherries, there are $2^{h_n-1} - p$ leaves of depth $h_n-1$ and $2p$ leaves of depth $h_n$. Thus, the minimum value is $(2^{h_n-1} - p) \cdot f(h_n-1) + 2 \cdot p \cdot f(h_n)$. This completes this part of the proof.

        \item Let $f$ be any increasing function. By Proposition \ref{Prop:leaf_depth_meta_star_min}, we need to show that $\Phi^{\Delta}_f\left(T^{star}_n\right) = n \cdot f(1)$. This holds true because all $n$ leaves of the star tree have depth $1$.
    \end{enumerate}
    This completes the proof.
\end{proof}

So far, we have thoroughly analyzed the three classes of metaconcepts, focusing on the trees that minimize and maximize them, as well as their minimum and maximum values. In the next section, we shift our attention to two additional properties a metaconcept can have: locality and recursiveness.

\subsubsection{Locality and recursiveness}
\label{Sec:locality_recursiveness}

In this section, we analyze the locality and prove the recursiveness of the metaconcepts, focusing again on the first-order metaconcepts. Unlike previous results, the conclusions of this section can generally not be extended to higher-order metaconcepts, not even to those that are equivalent to a first-order metaconcept. For example, \citet[Proposition 12.2 and Proposition 13]{Fischer2023} proved that the Colless index is local, but the equivalent corrected Colless index is not. Similarly, the Sackin index is local, whereas the equivalent average leaf depth is not (\citet[Proposition 5.4 and Proposition 6.3]{Fischer2023}). Further, the recursions for a first-order metaconcept do not account for additional values and, thus, do not apply to higher-order metaconcepts. However, it is worth noting that the corrected Colless index and the average leaf depth are recursive (\citet[Proposition 13.2 and Proposition 6.2]{Fischer2023}).

We begin with the locality. Based on the induced imbalance indices, one might conjecture that the BVM, the CSM, and the LDM are local, since all known imbalance indices induced by the corresponding first-order metaconcepts are local. These include the Colless index (\citet[Proposition 12.2]{Fischer2023}), the quadratic Colless index (\citet[Proposition 15.3]{Fischer2023}), the Sackin index (\citet[Proposition 5.4]{Fischer2023}), and the $\widehat{s}$-shape statistic (\citet[Proposition 9.3]{Fischer2023}). In the next proposition, we will show that this conjecture is true for all functions in the case of the BVM and CSM, and for affine functions in the case of the LDM.

\begin{Prop}\leavevmode

\label{Prop:locality}
    \begin{enumerate}
        \item The balance value metaconcept $\Phi^{\mathcal{B}}_f$ and the clade size metaconcept $\Phi^{\mathcal{N}}_f$ are local for all (not necessarily increasing) functions $f$.
        \item The leaf depth metaconcept $\Phi^{\Delta}_f$ is local if and only if $f$ is affine, i.e., $f(\delta) = m \cdot \delta + a$ with $m,a \in \R$.
    \end{enumerate}
\end{Prop}

Before proving this statement, we need a lemma that provides an equivalent condition for a function to be affine.

\begin{Lem}
\label{Lem:affine}
    Let $f: \R \rightarrow \R$ be a function. Then, $f$ is affine, i.e., $f(x) = m \cdot x + a$ with $m,a \in \R$, if and only if $f(x+z) - f(y+z) = f(x) - f(y)$ for all $x,y,z \in \R$.
\end{Lem}

The proof of this lemma can be found in Appendix \ref{appendix:proof-lemmas}. 
\medskip

Now, we are in a position to prove the proposition above.

\begin{proof}[Proof of Proposition \ref{Prop:locality}]\leavevmode
    \begin{enumerate}
        \item First, we consider the BVM and the CSM. The proof follows the same reasoning as \citet{Fischer2023} (proof of Proposition 12.2 and Proposition 5.4), where the locality of the Colless index and the Sackin index was established. The only difference is that, in our case, the summands are not merely the balance values or clade sizes but rather their evaluations under the function $f$. However, this modification does not affect the overall argument, thereby proving that both the BVM and the CSM are local.

        \item Now, we establish that the LDM is local if and only if $f$ is an affine function, i.e., $f(\delta) = m \cdot \delta + a$ with $m,a \in \R$.

        Consider a tree $T' \in \T$ obtained from a tree $T \in \T$ by replacing a subtree $T_v$ of $T$ with another subtree $T'_v$, where both subtrees have the same number of leaves and are rooted in $v$ with depth $\delta_{T}(v) = \delta_{T'}(v)$.
        Observe that the leaf sets outside these subtrees remain unchanged, i.e., $V_L(T) \setminus V_L(T_v) = V_L(T') \setminus V_L(T'_v)$, and for all $x \in V_L(T) \setminus V_L(T_v)$, we have $\delta_{T}(x) = \delta_{T'}(x)$, since modifying $T_v$ does not change the distance from the root to those leaves. Additionally, for leaves within these subtrees, we have $\delta_T(x) = \delta_T(v) + \delta_{T_v}(x)$ for $x \in V_L(T_v)$ and $\delta_{T'}(x) = \delta_{T'}(v) + \delta_{T'_v}(x)$ for $x \in V_L(T'_v)$. Now, we compute the difference
        \begin{align*}
        	\Phi^{\Delta}_{f}(T) - \Phi^{\Delta}_{f}(T') &=\quad \sum\limits_{x \in V_L(T_v)} f(\delta_T(x)) + \sum\limits_{x \in V_L(T) \setminus V_L(T_v)} f(\delta_T(x))\\
        	&\quad- \sum\limits_{x \in V_L(T'_v)} f(\delta_{T'}(x)) - \sum\limits_{x \in V_L(T') \setminus V_L(T'_v)} f(\delta_{T'}(x))\\
        	&=\quad \sum\limits_{x \in V_L(T_v)} f(\delta_T(v) + \delta_{T_v}(x)) + \sum\limits_{x \in V_L(T) \setminus V_L(T_v)} f(\delta_T(x))\\
        	&\quad- \sum\limits_{x \in V_L(T'_v)} f(\underbrace{\delta_{T'}(v)}_{= \delta_T(v)} + \delta_{T'_v}(x)) - \sum\limits_{x \in V_L(T) \setminus V_L(T_v)} f(\delta_{T}(x))\\
        	&\stackrel{\text{Lem. \ref{Lem:affine}}}{=} \sum\limits_{x \in V_L(T_v)} f(\delta_{T_v}(x)) - \sum\limits_{x \in V_L(T'_v)} f(\delta_{T'_v}(x)) = \Phi^{\Delta}_{f}(T_v) - \Phi^{\Delta}_{f}(T'_v)
        \end{align*}
        Thus, the LDM is local if and only if $f$ is affine, which completes the proof.
    \end{enumerate}
\end{proof}

A direct consequence of the locality of the metaconcepts is that every subtree of a tree minimizing (respectively, maximizing) a metaconcept, is itself a minimizing (respectively, maximizing) tree for the metaconcept.

Next, we establish the recursiveness of our metaconcepts.

\begin{Prop}
\label{Prop:recursiveness}
    The balance value metaconcept $\Phi^{\mathcal{B}}_f$ is a binary recursive tree shape statistic for all functions $f$. Similarly, the clade size metaconcept $\Phi^{\mathcal{N}}_f$ is a recursive tree shape statistic for all functions $f$. In contrast, the leaf depth metaconcept $\Phi^{\Delta}_f$ is a recursive tree shape statistic if $f$ is an affine function with an intercept of $a = 0$.

    Let $T^b \in \BT$ be a binary tree with standard decomposition $T^b = \left(T^b_1, T^b_2\right)$ such that the maximal pending subtree $T^b_i$ has $n^b_i$ leaves. Similarly, let $T \in \T$ be an arbitrary tree with standard decomposition $T = (T_1, \ldots, T_k)$, where each maximal pending subtree $T_i$ has $n_i$ leaves.

    Let $f$ be an arbitrary function, and let $f_m$ be an affine function with slope $m \in \R$ and intercept $a = 0$. If $n = 1$, we have $\Phi^{\mathcal{B}}_{f}\left(T^b\right) = 0$, $\Phi^{\mathcal{N}}_{f}(T) = 0$, $\Phi^{\Delta}_{f_m}(T) = f_m(0) = 0$. Moreover, for $n \geq 2$, we have:

    \begin{itemize}
        \item $\Phi^{\mathcal{B}}_{f}\left(T^b\right) = \Phi^{\mathcal{B}}_{f}\left(T^b_1\right) + \Phi^{\mathcal{B}}_{f}\left(T^b_2\right) + f\left(\left|n^b_1 - n^b_2\right|\right)$,

    	\item $\Phi^{\mathcal{N}}_{f}(T) = \sum\limits_{i=1}^{k} \Phi^{\mathcal{N}}_{f}(T_i) + f(n_1 + \ldots + n_k)$,

        \item $\Phi^{\Delta}_{f_m}(T) = \sum\limits_{i=1}^{k} \Phi^{\Delta}_{f_m}(T_i) + (n_1 + \ldots + n_k) \cdot m$.
    \end{itemize}
\end{Prop}
\begin{proof}
    Let $T = (T_1, \ldots, T_k)$, $T^b = \left(T^b_1,T^b_2\right)$, $n_i$, $n^b_i$, $f$, and $f_m$ as described.

    For $n \geq 2$, we can calculate the BVM as follows
    \begin{align*}
        \Phi^{\mathcal{B}}_{f}\left(T^b\right) &= \sum\limits_{b \in \mathcal{B}\left(T^b\right)} f(b) \stackrel{\text{Rem. \ref{Rem:recursiveness_sequences}}}{=} \sum\limits_{b \in \mathcal{B}\left(T^b_1\right)} f(b) + \sum\limits_{b \in \mathcal{B}\left(T^b_2\right)} f(b) + f\left(\left|n^b_1 - n^b_2\right|\right)\\
        &= \Phi^{\mathcal{B}}_{f}\left(T^b_1\right) + \Phi^{\mathcal{B}}_{f}\left(T^b_2\right) + f\left(\left|n^b_1 - n^b_2\right|\right).
    \end{align*}
    Thus, it is a binary recursive tree shape statistic of length $x = 2$, where the recursions $r_1$ and $r_2$ are:
    \begin{itemize}
        \item BVM: $\lambda_1 = \Phi^{\mathcal{B}}_{f}\left(T^{cat}_1\right) = 0$ and\\ $r_1\left(\left(r_1\left(T^b_1\right)\right),r_2\left(T^b_1\right)\right), \left(r_1\left(T^b_2\right),r_2\left(T^b_2\right)\right) = \Phi^{\mathcal{B}}_{f}\left(T^b_1\right) + \Phi^{\mathcal{B}}_{f}\left(T^b_2\right) + f\left(\left|n^b_1 - n^b_2\right|\right)$
        \item leaf number: $\lambda_2 = 1$ and $r_2\left(\left(r_1\left(T^b_1\right),r_2\left(T^b_1\right)\right), \left(r_1\left(T^b_2\right),r_2\left(T^b_2\right)\right)\right) = n^b_1 + n^b_2$
    \end{itemize}
    Hence, we have $\lambda \in \R^2$ and $r_i: \R^2 \times \R^2 \rightarrow \R$. Moreover, all recursions $r_i$ are symmetric. This completes the proof for the BVM.\\

    For $n \geq 2$, we can calculate the CSM as follows
    \begin{align*}
        \Phi^{\mathcal{N}}_{f}(T) &= \sum\limits_{n_v \in \mathcal{N}(T)} f(n_v) \stackrel{\text{Rem. \ref{Rem:recursiveness_sequences}}}{=} \sum\limits_{n_v \in \mathcal{N}(T_1)} f(n_v) + \ldots + \sum\limits_{n_v \in \mathcal{N}(T_k)} f(n_v) + f(n)\\
        &= \sum\limits_{i=1}^{k} \Phi^{\mathcal{N}}_{f}(T_i) + f(n_1 + \ldots + n_k).
    \end{align*}
    Thus, it is a recursive tree shape statistic of length $x = 2$, where the recursions $r_1$ and $r_2$ are:
    \begin{itemize}
        \item CSM: $\lambda_1 = \Phi^{\mathcal{N}}_{f}\left(T^{cat}_1\right) = 0$ and\\ $r_1\left(\left(r_1\left(T_1\right), r_2\left(T_1\right)\right), \ldots, \left(r_1\left(T_k\right), r_2\left(T_k\right)\right)\right) = \Phi^{\mathcal{N}}_{f}\left(T_1\right) + \ldots + \Phi^{\mathcal{N}}_{f}\left(T_k\right) + f\left(n_1 + \ldots + n_k\right)$,
        \item leaf number: $\lambda_2 = 1$ and $r_2\left(\left(r_1\left(T_1\right), r_2\left(T_1\right)\right), \ldots, \left(r_1\left(T_k\right), r_2\left(T_k\right)\right)\right) = n_1 + \ldots + n_k$.
    \end{itemize}
    Hence, we have $\lambda \in \R^2$ and $r_i: \underbrace{\R^2 \times \ldots \times \R^2}_{k \text{ times}} \rightarrow \R$. Moreover, all recursions $r_i$ are symmetric. This completes the proof for the CSM.\\

    For $n \geq 2$, we can calculate the LDM as follows
    \begin{align*}
        \Phi^{\Delta}_{f_m}\left(T\right) &= \sum\limits_{\delta \in \Delta\left(T\right)} f_m\left(\delta\right) \stackrel{\text{Rem. \ref{Rem:recursiveness_sequences}}}{=} \sum\limits_{\delta \in \Delta\left(T_1\right)} f_m\left(\delta +1\right) + \ldots + \sum\limits_{\delta \in \Delta\left(T_k\right)} f_m\left(\delta+1\right)\\
        &= \left(\sum\limits_{\delta \in \Delta\left(T_1\right)} f_m\left(\delta\right) + m\right) + \ldots + \left(\sum\limits_{\delta \in \Delta\left(T_k\right)} f_m\left(\delta\right) + m\right)\\
        &= \sum\limits_{i=1}^{k} \Phi^{\Delta}_{f_m}\left(T_i\right) + \left(n_1 + \ldots + n_k\right) \cdot m.
    \end{align*}
    Thus, it is a recursive tree shape statistic of length $x = 2$, where the recursions $r_1$ and $r_2$ are:
    \begin{itemize}
        \item LDM: $\lambda_1 = \Phi^{\Delta}_{f_m}\left(T^{cat}_1\right) = f_m\left(0\right) = 0$ and\\ $r_1\left(\left(r_1\left(T_1\right), r_2\left(T_1\right)\right), \ldots, \left(r_1\left(T_k\right), r_2\left(T_k\right)\right)\right) = \sum\limits_{i=1}^{k} \Phi^{\Delta}_{f_m}\left(T_i\right) + \left(n_1 + \ldots + n_k\right) \cdot m$,
        \item leaf number: $\lambda_2 = 1$ and $r_2\left(\left(r_1\left(T_1\right), r_2\left(T_1\right)\right), \ldots, \left(r_1\left(T_k\right), r_2\left(T_k\right)\right)\right) = n_1 + \ldots + n_k$.
    \end{itemize}
    Hence, we have $\lambda \in \R^2$ and $r_i: \underbrace{\R^2 \times \ldots \times \R^2}_{k \text{ times}} \rightarrow \R$. Moreover, all recursions $r_i$ are symmetric. This completes the proof for the LDM.
\end{proof}

\section{Discussion} \label{Sec:Discussion}
While tree balance is typically quantified using a single index function, \citet{Cleary2025} recently introduced a functional instead of a function to measure tree balance for rooted trees. This functional is based on the clade size sequence of a tree and depends on another function $f$. In this manuscript, we have generalized this concept to a broader framework, which we call the imbalance index metaconcept of order $\omega$. This metaconcept allows for any tree shape sequence as its underlying sequence. Exploiting this property, we introduced two additional subclasses alongside the clade size metaconcept (CSM): the balance value metaconcept (BVM) and the leaf depth metaconcept (LDM). We thoroughly analyzed these three metaconcepts with respect to their underlying function $f$. As a result, we identified many families of functions $f$ for which these metaconcepts yield (binary) imbalance indices, leading to a range of new imbalance indices. To help users identify a suitable imbalance index obtained from a metaconcept, we provided four decision trees. Additionally, we included \textsf{R} code for computing the metaconcepts and, consequently, the resulting imbalance indices. Furthermore, we analyzed the trees that maximize the metaconcepts for all leaf numbers, as well as the minimizing trees. Finally, we determined their minimum and maximum values.

Recall that the Sackin index and the Colless index are induced by the first-order CSM and BVM, respectively, when $f$ is the identity function, i.e., a strictly increasing and affine function. Consequently, both are minimized by several trees in $\BT$, including the gfb-tree and the mb-tree. We proved that when the identity function is approximated from above by a strictly increasing and strictly convex function, the mb-tree becomes the unique minimizer of both the CSM and the BVM. Conversely, when the identity function is approximated from below by a strictly increasing and strictly concave function, the gfb-tree is the unique minimizer of the CSM. Moreover, this result extends not only to the identity function but to all strictly increasing and affine functions. This is remarkable, because minor changes in the function $f$ lead to major changes in the set of minimizing trees. We elaborate on this observation in Figure~\ref{Fig:SackinApprox}.

Further, since many known imbalance indices also depend on one of the three sequences underlying the metaconcepts, we can classify these indices by determining which metaconcept they satisfy. We found that seven known imbalance indices fall into one of the three metaconcept classes, which can be further divided into several subclasses depending on the choice of $f$. Consequently, some of our results for the metaconcepts recover results from the literature in a more concise way. For example, Theorem \ref{Theo:bal_meta_imbalance_index} unifies six separate proofs from the literature into a single result. Additionally, one subclass, the LDM with strictly increasing and convex $f$, is totally new to the literature in the sense that none of the existing imbalance indices is induced by it.

A promising direction for future research is to further compare the introduced metaconcepts, for example, by analyzing the resolution of the imbalance indices they induce. It stands to reason that the LDM is the least resolved metaconcept, as for trees with eight leaves, there exist four pairs of trees that share the same $\Delta$ (leaf depth sequence), as well as a quartet of trees with identical $\Delta$ values, none of which share the same $\mathcal{N}$ (clade size sequence) or $\mathcal{B}$ (balance value sequence).

Another avenue for future work is to modify the imbalance index metaconcept itself. For instance, one could use sequences based on pairs of vertices (as in the original definition of the total cophenetic index), or consider sequences based on the edges of a tree rather than the vertices.

\begin{figure}
    \centering
    \includegraphics[scale=1]{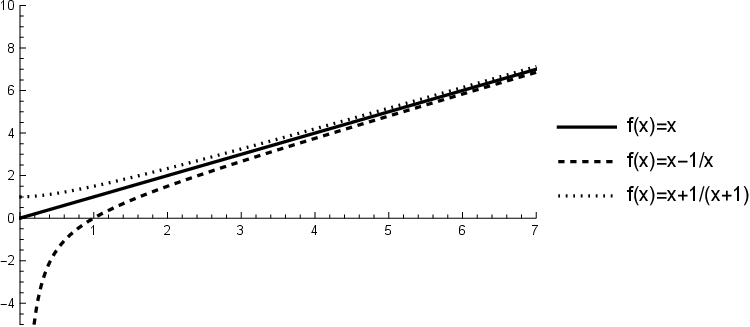}
    \caption{This figure shows (solid line) the strictly increasing and affine identity function $f(x) = x$, on which the Colless index and the Sackin index are based. In this case, the BVM $\Phi^{\mathcal{B}}_f$ and the CSM $\Phi^{\mathcal{N}}_f$ have several minima, including the gfb-tree and the mb-tree. It also shows (dashed line) a strictly increasing and strictly concave approximation of the identity from below, namely $f(x) = x - \frac{1}{x}$. Note that all functions of the type $f(x) = x - \frac{1}{x^a}$ for $a \geq 1$ are such approximations, and all of them will lead to the CSM $\Phi^{\mathcal{N}}_f$ having the gfb-tree as its unique minimum. The figure also shows (dotted line) a strictly increasing and strictly convex approximation of the identity from above, namely $f(x) = x + \frac{1}{(x+1)}$. Note that all functions of the type $f(x) = x + \frac{1}{(x+1)^a}$ for $a \geq 1$ are such approximations, and all of them will lead to the BVM $\Phi^{\mathcal{B}}_f$ and the CSM $\Phi^{\mathcal{N}}_f$ having the mb-tree as their unique minimum. Also note that we used $x+1$ in the last case, because using $x$ in the function would not lead to a strictly increasing $f$ for all passed values greater equal zero.}
    \label{Fig:SackinApprox}
\end{figure}

\section*{Acknowledgment} The authors wish to thank Volkmar Liebscher for various discussions and helpful insights. Parts of this material are based upon work supported by the National Science Foundation under Grant No.~DMS-1929284 while MF and KW were in residence at the Institute for Computational and Experimental Research in Mathematics in Providence, RI, during the Theory, Methods, and Applications of Quantitative Phylogenomics semester program.

\section*{Conflict of interest} The authors herewith certify that they have no affiliations with or involvement in any organization or entity with any financial (such as honoraria; educational grants; participation in speakers’ bureaus; membership, employment, consultancies, stock ownership, or other equity interest; and expert testimony or patent-licensing arrangements) or non-financial (such as personal or professional relationships, affiliations, knowledge or beliefs) interest in the subject matter discussed in this manuscript.

\section*{Data availability statement} 
Data sharing is not applicable to this article as no new data were created or analyzed in this study.

\section*{Authors' contributions} All authors contributed equally.

\bibliographystyle{plainnat}
\bibliography{References_metaconcepts}

\clearpage
\begin{appendices}

\section{Additional figures} \label{appendix:figures}
All inner vertices $v$ of all trees in Figure \ref{Fig:6_5_6}-\ref{Fig:13_869_957} are labeled by a pair $(b_v,n_v)$ containing its balance value $b_v$ and its clade size $n_v$. If the leaves are labeled, then they are labeled by their leaf depth.

\begin{figure}[htbp]
    \centering
    \includegraphics[scale=2]{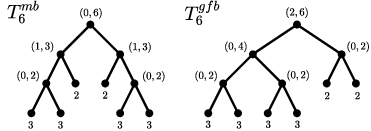}
    \caption{Unique minimal example of two binary trees having the same $\Delta$ but different $\mathcal{B}$ and $\mathcal{N}$, respectively. Specifically, $n=6$, $\mathcal{B}\left(T^{mb}_6\right) = (0,0,0,1,1) \neq (0,0,0,0,2) = \mathcal{B}\left(T^{gfb}_6\right)$, $\mathcal{N}\left(T^{mb}_6\right) = (2,2,3,3,6) \neq (2,2,2,4,6) = \mathcal{N}\left(T^{gfb}_6\right)$, and $\Delta\left(T^{mb}_6\right) = \Delta\left(T^{gfb}_6\right) = (2,2,3,3,3,3)$.}
    \label{Fig:6_5_6}
\end{figure}

\begin{figure}[htbp]
    \centering
    \includegraphics[width=\textwidth]{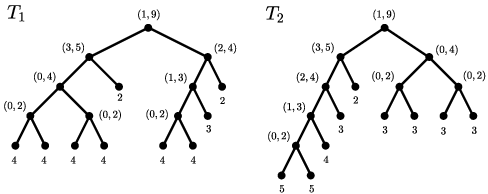}
    \caption{Unique minimal example of two binary trees having the same $\mathcal{N}$ and the same $\mathcal{B}$. Note, however, that $\Delta$ is not the same. Specifically, $n = 9$, $\mathcal{B}(T_1) = \mathcal{B}(T_2) = (0,0,0,0,1,1,2,3)$, $\mathcal{N}(T_1) = \mathcal{N}(T_2) = (2,2,2,3,4,4,5,9)$, and $\Delta(T_1) = (2,2,3,4,4,4,4,4,4) \neq (2,3,3,3,3,3,4,5,5) = \Delta(T_2)$.}
    \label{Fig:9_42_44}
\end{figure}

\begin{figure}[htbp]
    \centering
    \includegraphics[width=\textwidth]{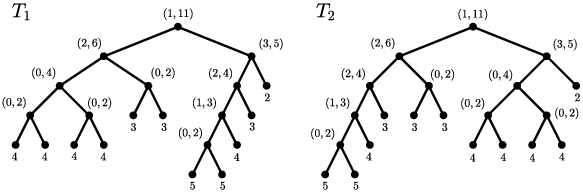}
    \caption{Unique minimal example of two binary trees having the same $\mathcal{B}$, the same $\mathcal{N}$, and the same $\Delta$. Specifically, $n = 11$, $\mathcal{B}(T_1) = \mathcal{B}(T_2) = (0,0,0,0,0,1,1,2,2,3)$, $\mathcal{N}(T_1) = \mathcal{N}(T_2) = (2,2,2,2,3,4,4,5,6,11)$, and $\Delta(T_1) = \Delta(T_2) = (2,3,3,3,4,4,4,4,4,5,5)$}
    \label{Fig:11_194_199}
\end{figure}

\begin{figure}[htbp]
    \centering
    \includegraphics[width=\textwidth]{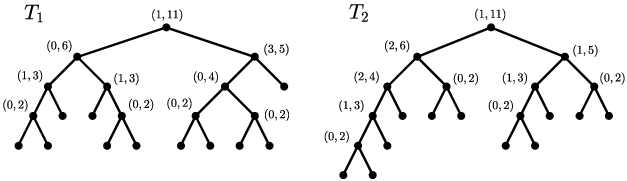}
    \caption{Unique minimal example of two binary trees having different $\mathcal{B}$ but the same $\mathcal{N}$. Specifically, $n = 11$, $\mathcal{B}(T_1) = (0,0,0,0,0,0,1,1,1,3) \neq (0,0,0,0,1,1,1,1,2,2) = \mathcal{B}(T_2)$, and $\mathcal{N}(T_1) = \mathcal{N}(T_2) = (2,2,2,2,3,3,4,5,6,11)$.}
    \label{Fig:11_201_205}
\end{figure}

\begin{figure}[htbp]
    \centering
    \includegraphics[width=\textwidth]{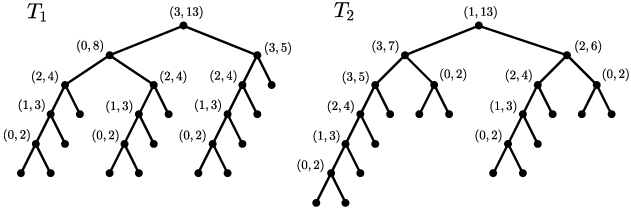}
    \caption{One of 13 minimal examples of two binary trees having the same $\mathcal{B}$ but different $\mathcal{N}$. Here, $n = 13$, $\mathcal{B}(T_1) = \mathcal{B}(T_2) = (0,0,0,0,1,1,1,2,2,2,3,3)$, and $\mathcal{N}(T_1) = (2,2,2,3,3,3,4,4,4,5,8,13) \neq (2,2,2,2,3,3,4,4,5,6,7,13) = \mathcal{N}(T_2)$.}
    \label{Fig:13_869_957}
\end{figure}

\clearpage
\section{Proofs of additional lemmas}
\label{appendix:proof-lemmas}

\begin{proof}[Proof of Lemma~\ref{Lem:number_subtrees_size_mb}] 
    The case $i = 1$ is trivial. Thus, we subsequently assume $n \geq i \geq 2$. Before we consider these cases more in-depth, we start with some general observations. Throughout the proof, we subdivide the vertices in $T^{mb}_n$ into layers, where the root is the only vertex in layer 0, the root's children are the only vertices in layer 1 and so forth. Then the root of a subtree $T_i$ of $T^{mb}_n$ with $i$ leaves can only be located in layers $h_n - h_i-1$ or $h_n - h_i$. This is because all leaves must either be contained in layer $h_n$ or $h_n-1$ as $T^{mb}_n$ has height $h_n$ (cf. Remark \ref{Rem:gfb_mb_from_fb_min_Sackin}), and as all of its pending subtrees are also maximally balanced trees (cf. Section \ref{sec:def}), which implies that they have height $h_i$.

    The high-level idea of the proof now is as follows:

    \begin{itemize}
        \item We show that each layer $l$ contains up to $2^l$ vertices, and each layer $l$, possibly except for layer $h_n$ in case $n < 2^{h_n}$, contains precisely $2^l$ vertices each.
        \item We show that each vertex in layer $l$ is the root of a subtree of size either $\left\lfloor\frac{n}{2^l}\right\rfloor$ or $\left\lceil\frac{n}{2^l}\right\rceil$. Thus, no other subtree sizes are possible.
        \item For $i > 1$, we show that only subtrees of size $i = 2$ can be contained in more than one layer. All other subtree sizes can occur in only one layer (i.e., for $i \geq 3$, if there are subtrees of size $i$, they are all rooted in the same layer).
        \item We then count in each layer the subtrees of sizes $\left\lfloor\frac{n}{2^l}\right\rfloor$ and $\left\lceil\frac{n}{2^l}\right\rceil$, respectively.
    \end{itemize}

    We start with proving that layer $l$ in $T^{mb}_n$ contains up to $2^l$ vertices, all of which induce subtrees with either $\left\lfloor\frac{n}{2^l}\right\rfloor$ or $\left\lceil\frac{n}{2^l} \right\rceil$ many leaves. We prove this by induction on $l$. For $l = 0$, we have only the root, so indeed we have $2^0 = 1$ vertices in this layer, and this vertex is the root of the entire tree, so of $n = \frac{n}{2^0}$ leaves. This completes the base case. For the inductive step, let us assume we know that the statement holds up to $l$ and now consider layer $l+1$. It is clear that as layer $l$ has at most $2^l$ vertices by induction, layer $l+1$ can have at most $2^{l+1}$ vertices, because each vertex in layer $l$ has at most two children in layer $l+1$. Moreover, each vertex $v$ in layer $l+1$ is the child of a vertex $u$ in layer $l$, which by induction is the root of a tree $T_u$ with either $\left\lfloor\frac{n}{2^l}\right\rfloor$ or $\left\lceil\frac{n}{2^l}\right\rceil$ many leaves. Thus, we know by the definition of $T^{mb}_n$ that this only leaves four options for $v$: If $T_u$ has $\left\lfloor\frac{n}{2^l}\right\rfloor$ many leaves, then the tree $T_v$ induced by $v$ can only have either $\left\lfloor\frac{\left\lfloor\frac{n}{2^l}\right\rfloor}{2}\right\rfloor$ or $\left\lceil\frac{\left\lfloor\frac{n}{2^l}\right\rfloor}{2}\right\rceil$ leaves, and if $T_u$ has $\left\lceil\frac{n}{2^l}\right\rceil$ many leaves, then the tree $T_v$ induced by $v$ can only have either $\left\lfloor\frac{\left\lceil\frac{n}{2^l}\right\rceil}{2}\right\rfloor$ or $\left\lceil\frac{\left\lceil\frac{n}{2^l}\right\rceil}{2} \right\rceil$ leaves. However, using the well-known identities $\left\lfloor\frac{\left\lfloor \frac{a}{b}\right\rfloor}{c}\right\rfloor = \left\lfloor\frac{a}{bc}\right\rfloor$ and $\left\lceil\frac{\left\lceil\frac{a}{b}\right\rceil}{c}\right\rceil = \left\lceil\frac{a}{bc} \right\rceil$ (which hold for all real numbers $a, b$ as well as positive integers $c$), we note that $\left\lfloor\frac{\left\lfloor\frac{n}{2^l}\right\rfloor}{2}\right\rfloor = \left\lfloor\frac{n}{2^{l+1}}\right\rfloor$ and $\left\lceil\frac{\left\lceil\frac{n}{2^l}\right\rceil}{2}\right\rceil = \left\lceil\frac{n}{2^{l+1}}\right\rceil$ as desired.
    Similarly, for the remaining two possible values, before using the same identities, we additionally have to convert the inner floor and ceiling functions first by using the identities $\left\lfloor\frac{a}{b}\right\rfloor = \left\lceil\frac{a-b+1}{b}\right\rceil$ and $\left\lceil\frac{a}{b}\right\rceil = \left\lfloor\frac{a+b-1}{b}\right\rfloor$ (which hold for all real numbers $a$ and positive numbers $b$). This gives us $\left\lceil\frac{\left\lfloor\frac{n}{2^l}\right\rfloor}{2}\right\rceil = \left\lceil\frac{\left\lceil\frac{n-2^l+1}{2^l}\right\rceil}{2}\right\rceil = \left\lceil\frac{n-2^l+1}{2^{l+1}}\right\rceil = \left\lfloor\frac{\left(n-2^l+1\right)+2^{l+1}-1}{2^{l+1}}\right\rfloor = \left\lfloor\frac{n}{2^{l+1}}+\frac{1}{2}\right\rfloor \in \left\{\left\lfloor\frac{n}{2^{l+1}}\right\rfloor, \left\lceil\frac{n}{2^{l+1}}\right\rceil\right\}$ as well as $\left\lfloor\frac{\left\lceil\frac{n}{2^l}\right\rceil}{2}\right\rfloor = \left\lfloor\frac{\left\lfloor\frac{n+2^l-1}{2^l}\right\rfloor}{2}\right\rfloor = \left\lfloor\frac{n+2^l-1}{2^{l+1}}\right\rfloor = \left\lceil\frac{\left(n+2^l-1\right) - 2^{l+1}+1}{2^{l+1}}\right\rceil = \left\lceil\frac{n}{2^{l+1}} - \frac{1}{2}\right\rceil \in \left\{\left\lfloor\frac{n}{2^{l+1}}\right\rfloor, \left\lceil\frac{n}{2^{l+1}}\right\rceil\right\}$. Therefore, in all cases, we have that the number $n_v$ of leaves of $T_v$ is either $\left\lfloor\frac{n}{2^{l+1}}\right\rfloor$ or $\left\lceil\frac{n}{2^{l+1}}\right\rceil$, which completes the induction.

    Moreover, note that in $T^{mb}_n$, all layers $l$ except possibly for layer $h_n$ are \enquote{full} in the sense that they contain precisely $2^l$ many vertices. This must be true as otherwise there would be two leaves with a depth difference of more than $1$, which is not possible in $T^{mb}_n$ (cf. Remark \ref{Rem:gfb_mb_from_fb_min_Sackin}). Also note that layer $h_n$ only contains leaves, i.e., it can only induce subtrees of size $1$.

    Before we continue, we now consider the possible layers $l$ of $T^{mb}_n$ in which a subtree of size $i$ can be rooted. For $i = 2$, there must be at least one such subtree (as $n \geq 2$), and this is necessarily rooted in layer $h_n-1$ (as the parent of the cherry with maximal depth induces such a subtree). However, depending on $n$, $T^{mb}_n$ can also contain cherry parents on layer $h_n-2$ (as an example, consider $T_5^{mb}$ depicted as tree $T$ in Figure \ref{Fig:cherry_relocation}). But there cannot be such a cherry parent on layer $h_n-3$ or lower, because this would induce leaves on layer $h_n-2$ or lower, which would inevitably lead to two leaves of a depth difference of more than one in $T^{mb}_n$ (between a leaf of maximum depth $h_n$ and the newly found leaf on layer $h_n-2$ or smaller), which cannot happen in $T^{mb}_n$ (cf. Remark \ref{Rem:gfb_mb_from_fb_min_Sackin}).

    Now, consider $i \geq 3$. We show that in this case, the layer $l$ with $l \leq h_n-2$ of $T^{mb}_n$ in which all subtrees of size $i$ can potentially be contained is uniquely determined. In particular, it is $\left\lfloor\log_2\left(\frac{n}{i}\right)\right\rfloor$ or $\left\lceil\log_2\left(\frac{n}{i}\right)\right\rceil$, depending on $i$. This is because if a subtree of size $i \geq 3$ is rooted in layer $l$, we already know that $i = \left\lfloor\frac{n}{2^l}\right\rfloor$ or $i = \left\lceil\frac{n}{2^l}\right\rceil$. In the first case, we have $ i \leq \frac{n}{2^l} < i+1$, and thus $2^l \cdot i \leq n < 2^l(i+1)$, which directly implies $l \leq \log_2\left(\frac{n}{i}\right)$, and thus, as $l$ is an integer, $l \leq \left\lfloor\log_2\left(\frac{n}{i}\right)\right\rfloor$. Moreover, from $2^l \cdot i \leq n < 2^l(i+1)$ we also derive $\log_2\left(\frac{n}{i+1}\right) < l$. Now assume $l < \left\lfloor\log_2\left(\frac{n}{i}\right)\right\rfloor$. Then, as $l$ is an integer, we have $\log_2\left(\frac{n}{i+1}\right) < l < l+1 \leq \left\lfloor\log_2\left(\frac{n}{i}\right)\right\rfloor\leq\log_2\left(\frac{n}{i}\right) $, which implies $\log_2\left(\frac{n}{i}\right) - \log_2\left(\frac{n}{i+1}\right) > 1$. However, this is a contradiction as $\log_2\left(\frac{n}{i}\right) - \log_2\left(\frac{n}{i+1}\right) = \log_2\left(\frac{i+1}{i}\right) < 1$ for all $i \geq 3$. Thus, we must have $l = \left\lfloor\log_2\left(\frac{n}{i}\right)\right\rfloor$ in the first case as desired.

    For the second case, analogously to the first case, it can be easily seen that we have $l = \left\lceil\log_2\left(\frac{n}{i}\right)\right\rceil$ as desired.

    Next, we complete our analysis of the case $i \geq 3$ by counting the subtrees of $T^{mb}_n$ of each such size. As we have seen that for $i \geq 3$ all subtrees of such a size must fulfill $i = \left\lfloor\frac{n}{2^l}\right\rfloor$ with $l = \left\lfloor\log_2\left( \frac{n}{i}\right)\right\rfloor$ or $i = \left\lceil\frac{n}{2^l}\right\rceil$ with $l = \left\lceil\log_2\left(\frac{n}{i}\right)\right\rceil$, it immediately follows that if $i$ does not meet this requirement, we have $mb_n(i) = 0$, thus proving the fifth statement of the theorem.

    Now consider any layer $l \leq h_n-2$ (note that larger layers cannot contain the root of a subtree of size 3 as $T^{mb}_n$ has height $h_n$). We have already seen that this layer contains precisely $2^l$ many vertices, each of which induces a subtree of size $i = \left\lfloor\frac{n}{2^l}\right\rfloor$ or $i = \left\lceil\frac{n}{2^l}\right\rceil$. We now distinguish two cases. First, assume $\left\lfloor\frac{n}{2^l}\right\rfloor = \left\lceil\frac{n}{2^l}\right\rceil = \frac{n}{2^l} \in \N$. Then, all $2^l$ many subtrees rooted in layer $l$ are of this size. This partially proves the seventh statement of the theorem, because in this case we have $r_l^n = 0$, which can be seen as follows:

    \begin{align*}
        \frac{n}{\left\lfloor\frac{n}{2^l}\right\rfloor} &= \frac{n}{\frac{n}{2^l}} = 2^l \hspace{0.3cm}
        \Rightarrow r_l^n = n-2^l \cdot \left\lfloor\frac{n}{2^l}\right\rfloor = n-\frac{n}{\left\lfloor\frac{n}{2^l}\right\rfloor} \cdot \left\lfloor\frac{n}{2^l}\right\rfloor = n-n = 0.
    \end{align*}

    If, however, $\left\lfloor\frac{n}{2^l}\right\rfloor < \left\lceil\frac{n}{2^l}\right\rceil$, informally speaking, the only way to divide the $n$ leaves between the $2^l$ subtrees of sizes $\left\lfloor\frac{n}{2^l}\right\rfloor$ and $\left\lceil\frac{n}{2^l}\right\rceil$ is to fill up all of the subtrees with $\left\lfloor\frac{n}{2^l}\right\rfloor$ many leaves and then add 1 leaf each to some of them, until the $r_l^n = n-2^l \cdot \left\lfloor\frac{n}{2^l}\right\rfloor$ leaves are used up. This leads to $r_l^n$ subtrees of size $\left\lceil\frac{n}{2^l}\right\rceil$ and to $2^l-r_l^n$ subtrees of size $\left\lfloor\frac{n}{2^l}\right\rfloor$, which completes the proof of the sixth and seventh cases.

    It remains to consider the case $i = 2$. As subtrees of size 2 are, besides the ones of size 1, the only ones that can occur on more than one layer, we have to distinguish three cases. Note that $mb_n(2) = 0$ is not possible as $n \geq 2$, so there must be at least one cherry. However, there are three cases: All cherries are rooted in layer $h_n-1$, and this layer may or may not also contain leaves, or both layers $h_n-1$ and $h_n-2$ contain cherries. In case all cherries are rooted in layer $h_n-1$ and this layer does \textit{not} contain any leaves, obviously all $2^{h_n-1}$ vertices in this layer are parents of a cherry, so $mb_n(i) = 2^{h_n-1}$. This, however, implies that layer $2^{h_n}$ contains $2^{h_n}$ many leaves, i.e., $n = 2^{h_n}$. This proves the second statement of the theorem.

    So from now on, we consider the case $n < 2^{h_n}$ and $i = 2$ to prove the third and fourth statement of the theorem. We have already seen that then not all vertices in layer $h_n-1$ can be roots of cherries, so layer $h_n-1$ must contain some leaves. It remains to distinguish the case in which all cherry parents of $T^{mb}_n$ are contained in layer $h_n-1$ from the case in which some of the cherry parents are in layer $h_n-2$.

    As we have seen, in both cases we have that in layer $h_n-1$, we have both cherry parents and leaves, so we must have $1 = \left\lfloor\frac{n}{2^{h_n-1}}\right\rfloor$ and $i = 2 = \left\lceil\frac{n}{2^{h_n-1}}\right\rceil$. Now, if cherry parents are also contained in layer $h_n-2$, they must there be the vertices inducing smaller subtrees, because all parents of the cherries in layer $h_n-1$ are also contained in layer $h_n-2$ and they induce larger subtrees. Thus, we must also have $i = 2 = \left\lfloor\frac{n}{2^{h_n-2}}\right\rfloor$. So in summary, if vertices inducing subtrees of size $i = 2$ are contained in both layers $h_n-1$ and $h_n-2$, we must have $i = 2 = \left\lfloor\frac{n}{2^{h_n-2}}\right\rfloor = \left\lceil\frac{n}{2^{h_n-1}}\right\rceil$. This corresponds to the condition of the fourth statement of the theorem. We now count the number of cherry parents in each of the two layers by the same arguments as in the case $i \geq 3$: All $2^{h_n-1}$ vertices in layer $h_n-1$ are roots of subtrees of size at least 1. The remaining $r_{h_n-1}^n = n-2^{h_n-1} \cdot\left\lfloor\frac{n}{2^{h_n-1}}\right\rfloor$ subtrees contain two leaves. This number needs to be added to the 2-leaf subtrees induced by layer $h_n-2$. In this layer, we know that all $2^{h_n-2}$ vertices induce trees of size at least $2$, but $r_{h_n-2}^n$ of them induce three leaves. So layer $h_n-2$ comes with $2^{h_n-2} - r_{h_n-2}^n$ many subtrees of size 2. Thus, layers $h_n-1$ and $h_n-2$ in this case together induce $2^{h_n-2} - r_{h_n-2}^n + r_{h_n-1}^n$ many subtrees of size 2, which proves the fourth statement of the theorem.

    So only the third statement of the theorem remains. We are still in the case where $n < 2^{h_n}$, and we now have that only layer $h_n-1$ contains cherry parents. As we have already seen, the other vertices in layer $h_n-1$ must be leaves, so we have $i = 2 = \left\lceil\frac{n}{2^{h_n-1}}\right\rceil$. As layer $h_n-1$ contains precisely $2^{h_n-1}$ vertices, by the same arguments as used above, we can conclude that the only way to distribute $n$ leaves to these vertices is to consider $2^{h_n-1}$ trees of size 1, i.e., single leaves, and to turn $r_{h_n-1}^n = n-2^{h_n-1} \cdot \left\lfloor\frac{n}{2^{h_n-1}}\right\rfloor$ many of these leaves into cherry parents. This leads to $r_{h_n-1}^n$ many cherries, which completes the third statement of the theorem and thus the entire proof.
\end{proof}

\begin{proof}[Proof of Lemma \ref{Lem:affine}]\leavevmode
\begin{itemize}
    \item [\enquote{$\Rightarrow$}]
    Let $f: \R \rightarrow \R$ be affine and $x,y,z \in \R$. Then,
    \[f(x+z) - f(y+z) = m \cdot(x+z) + a - (m \cdot(y+z) + a) = m \cdot x + a - (m \cdot y + a) = f(x) - f(y).\]
    \item [\enquote{$\Leftarrow$}]
    Let $f: \R \rightarrow \R$ be a function such that $f(x+z) - f(y+z) = f(x) - f(y)$ for all $x,y,z \in \R$. Then, for $x = 0$ and $z = 1$, we have
    \begin{align*}
        f(0+1) - f(y+1) &= f(0) - f(y) \, \text{ for all } y \in \R\\
        \Longleftrightarrow f(y+1) - f(y) &= f(1) - f(0) \, \text{ for all } y \in \R.
    \end{align*}
    Thus, the slope between $y$ and $y+1$ is constant, namely
    \[m = \frac{f(y+1) - f(y)}{(y+1) - y} = \frac{f(1) - f(0)}{1} = f(1) - f(0).\]
    Since $y \in \R$ was arbitrary, this implies that $f$ is affine, which completes the proof.
\end{itemize}
\end{proof}
\end{appendices}

\end{document}